\documentclass[10pt,twocolumn]{IEEEtran}
\usepackage{cite}
\usepackage{amsmath,amssymb,amsfonts,acronym}
\usepackage{amsthm}
\usepackage[ruled,vlined]{algorithm2e}
\usepackage{algorithmic}
\usepackage{textcomp}
\usepackage[usenames]{color}
\usepackage{psfrag}
\usepackage{arydshln}
\usepackage{bm}
\usepackage{url}
\usepackage{soul}
\usepackage{comment}
\usepackage{graphicx,color}

\usepackage{epstopdf}
\usepackage{esint}

\acrodef{3D}{three dimensional}
\acrodef{AP}{access point}
\acrodef{AoD}{angle-of-departure}
\acrodef{CoA}{curvature of arrival}
\acrodef{D2D}{device-to-device}
\acrodef{DFT}{discrete Fourier transform}
\acrodef{MPC}{multipath component}
\acrodef{TDL}{time delay line}

\acrodef{2D}{two dimensional}
\acrodef{OEB}{orientation error bound}
\acrodef{PEB}{Position Error Bound}

\acrodef{EKF}{Extended Kalman Filter}

\acrodef{5G}{fifth generation}
\acrodef{6G}{sixth generation}

\acrodef{TX}{transmitter}
\acrodef{RX}{receiver}

\acrodef{AF}{ambiguity function}

\acrodef{FFT}{fast Fourier transform}
\acrodef{TTD}{true-time delay}

\acrodef{CR}{channel response}

\acrodef{MAP}{maximum a posteriori probability}

\acrodef{MIMO}{multiple-input multiple-output}
\acrodef{mmW}{millimeter-wave}

\acrodef{SIMO}{single-input multiple-output}

\acrodef{MISO}{multiple-input single-output}
\acrodef{GNSS}{Global Navigation Satellite System}

\acrodef{SISO}{single-input single-output}
\acrodef{TDMA}{time division multiple access}
\acrodef{FDMA}{frequency division multiple access}
\acrodef{CRLB}{Cram\'er-Rao Lower Bound}
\acrodef{P-CRLB}{posterior Cram\'er-Rao Lower Bound}
\acrodef{SCAs}{small cell access points}

\acrodef{SCA}{small cell access point}

\acrodef{BS}{base station}

\acrodef{DF}{detect \& forward}

\acrodef{JF}{just forward}

\acrodef{DS}{delay spread}

\acrodef{CSI}{channel state information}

\acrodef{ACK}{acknowledge}
\acrodef{ADC}{analog-to-digital converter}
\acrodef{AWGN}{additive white Gaussian noise}
\acrodef{BPZF}{band-pass zonal filter}
\acrodef{CDF}{cumulative distribution function}
\acrodef{ch.f.}{characteristic function}
\acrodef{CIR}{channel impulse response}
\acrodef{CRB}{Cram\'{e}r-Rao bound}
\acrodef{DP}{direct path}
\acrodef{ED}{energy detector}
\acrodef{EM}{electromagnetic}
\acrodef{FCC}{Federal Communications Commission}
\acrodef{FIM}{Fisher Information Matrix}
\acrodef{GDOP}{geometric dilution of precision}
\acrodef{GLRT}{generalized likelihood ratio test}
\acrodef{GPS}{Global Positioning System}
\acrodef{INR}{interference-to-noise ratio}
\acrodef{IR-UWB}{impulse radio UWB}
\acrodef{i.i.d.}{independent, identically distributed}
\acrodef{IS}{importance sampling}
\acrodef{LF}{likelihood function}
\acrodef{LRT}{likelihood ratio test}
\acrodef{LIS}{large intelligent surface}
\acrodef{LOS}{line-of-sight}
\acrodef{LLR}{log-likelihood ratio}
\acrodef{LRT}{likelihood ratio test}
\acrodef{LS}{least squares}
\acrodef{MAC}{medium access control}
\acrodef{MF}{matched filter}
\acrodef{MLE}{maximum likelihood estimator}
\acrodef{mm-wave}{millimeter-waves}
\acrodef{MMSE}{minimum-mean-square-error}
\acrodef{MSE}{mean squared error}
\acrodef{MU}{multi-user}
\acrodef{MUI}{multi-user interference}
\acrodef{MUR}{Multistatic RADAR}
\acrodef{NBI}{narrowband interference}
\acrodef{NL}{nonlinear}
\acrodef{NLOS}{non-line-of-sight}
\acrodef{OAM}{Orbital Angular Momentum}
\acrodef{PAM}{pulse amplitude modulation}
\acrodef{PEB}{position error bound}
\acrodef{PDF}{probability distribution function}
\acrodef{TDOA}{time difference of arrival}
\acrodef{PDP}{power delay profile}
\acrodef{PFA}{probability of false alarm}
\acrodef{ppm}{part-per-million}
\acrodef{PPM}{pulse position modulation}
\acrodef{RFID}{radio frequency identification}
\acrodef{RMSE}{root mean square error}
\acrodef{RSS}{received signal strength}
\acrodef{RSSI}{received signal strength indicator}
\acrodef{RTT}{round-trip time}
\acrodef{RV}{random variable}
\acrodef{SIR}{sampling-importance-resampling}
\acrodef{SNR}{signal-to-noise ratio}
\acrodef{SU}{single-user}
\acrodef{TH}{time-hopping}
\acrodef{THz}{terahertz}
\acrodef{TNR}{threshold-to-noise ratio}
\acrodef{UAV}{unmanned aerial vehicle}
\acrodef{UCA}{uniform circular array}
\acrodef{ULA}{uniform linear array}
\acrodef{UWB}{ultrawide bandwidth}
\acrodef{WP}{waypoint}
\acrodef{WBI}{wideband interference}
\acrodef{WPAN}{wireless personal area networks}
\acrodef{WSN}{Wireless Sensor Network}
\acrodef{WWLB}{Weiss-Weinstein lower bound}

\acrodef{CW}{continuous wave}

\acrodef{RF}{radiofrequency}

\acrodef{FCC}{Federal Communications Commission}

\acrodef{EIRP}{effective radiated isotropic power}

\acrodef{RCS}{radar cross section}

\acrodef{BAV}{balanced antipodal Vivaldi}

\acrodef{PRake}{partial Rake}

\acrodef{RTLS}{Real-time locating systems}

\acrodef{EFI}{equivalent Fisher information matrix}

\acrodef{SPEB}{squared position error bound}

\acrodef{SOEB}{squared orientation error bound}

\acrodef{Hi-RADIAL}{High-accuracy RAdio Detection, Identification,
And Localization}

\acrodef{HCRB}{hybrid Cram\'{e}r-Rao bound}
\acrodef{HFIM}{hybrid Fisher Information Matrix}

\acrodef{ZZB}{Ziv-Zakai bound}

\acrodef{TOA}{time-of-arrival}
\acrodef{DOA}{direction-of-arrival}

\acrodef{ToF}{time-of-flight}

\acrodef{WSN}{wireless sensor network}

\acrodef{MAC}{medium access control}

\acrodef{RSS}{received signal strength}

\acrodef{TDoA}{time difference-of-arrival}

\acrodef{RF}{radiofrequency}
\acrodef{PF}{particle filter}
\acrodef{PSD}{power spectral density}

\acrodef{RTT}{round-trip time}

\acrodef{AOA}{angle-of-arrival}

\acrodef{MF}{matched filter}

\acrodef{ED}{energy detector}

\acrodef{ML}{maximum likelihood}

\acrodef{MUR}{Multistatic radar}

\acrodef{HDSA}{high-definition situation-aware}

\acrodef{RRC}{root raised cosine}

\acrodef{OFDM}{orthogonal frequency division multiplexing}

\acrodef{IF}{intermediate frequency}

\acrodef{PHY}{physical layer}

\acrodef{S-V}{Saleh-Valenzuela}

\acrodef{UHF}{ultra-high frequency}

\acrodef{PR}{pseudo-random}

\acrodef{SoC}{System on Chip}

\acrodef{SoP}{System on Package}

\acrodef{SPMF}{Single-Path Matched Filter}

\acrodef{IMF}{Ideal Matched Filter}

\acrodef{SCR}{signal-to-clutter ratio}

\acrodef{BEP}{bit error probability}

\acrodef{BER}{bit error rate}

\acrodef{WSR}{wireless sensor radar}

\acrodef{HPBW}{half power beam width}

\acrodef{LEO}{localization error outage}

\acrodef{SLAM}{simultaneous localization and mapping}
\acrodef{std}{standard deviation}
\acrodef{WPT}{wireless power transfer}
\acrodef{EV}{electric vehicle}

\newcommand{\ny}{n_{\mathsf{y}}}
\newcommand{\nz}{n_{\mathsf{z}}}
\newcommand{\ntilde}{\tilde{n}}
\newcommand{\Ntilde}{\tilde{N}}

\newcommand{\sqrtJd}{\sqrt{\left[\tjM\left(d\right) \right]^{-1}}}

\newcommand{\JM}{\mathbf{J}^{\mathrm{D}}}
\newcommand{\tJM}{\tilde{\mathbf{J}}^{\mathrm{D}}}
\newcommand{\tjM}{\tilde{{J}}^{\mathrm{D}}}
\newcommand{\JP}{\mathbf{J}^{\mathrm{P}}}

\newcommand{\lambdap}{\lambda}

\newcommand{\fn}{{f}_n}
\newcommand{\gn}{g_n}

\newcommand{\fdtk}{{f}_{n,k}}

\newcommand{\sumn}{\sum_{n=0}^{N-1}}

\newcommand{\fc} {f_p}

\newcommand{\Ns} {N_{\mathrm{s}}}

\newcommand{\ho}{h_0}
\newcommand{\hn}{h_n}
\newcommand{\hN}{h_{N-1}}

\newcommand{\pn}{\mathbf{q}_{n}}
\newcommand{\pnk}{\mathbf{q}_{n}}
\newcommand{\pok}{\mathbf{q}_{0}}

\newcommand{\xn}{x_{n}}
\newcommand{\yn}{y_{n}}
\newcommand{\zn}{z_{n}}

\newcommand{\xo}{x_{0}}
\newcommand{\yo}{y_{0}}
\newcommand{\zo}{z_{0}}

\newcommand{\dk}{d_k}

\newcommand{\dnk}{d_{n,k}}
\newcommand{\dn}{d_n}

\newcommand{\dno}{d_{n0}}
\newcommand{\thetak}{\theta_k}
\newcommand{\phik}{\phi_k}

\newcommand{\phino}{\phi_{n0}}
\newcommand{\thetano}{\theta_{n0}}

\newcommand{\pk}{\mathbf{p}_{k}}

\newcommand{\ssk}{\mathbf{s}_{k}}
\newcommand{\ssztok}{\mathbf{s}_{0:k}}
\newcommand{\sskm}{\mathbf{s}_{m,k}}
\newcommand{\sszm}{\mathbf{s}_{m,0}} 
 \newcommand{\ssz}{\mathbf{s}_0}
\newcommand{\ssom}{\mathbf{s}_{m,1}} 
\newcommand{\sskmo}{\mathbf{s}_{m, k-1}}
\newcommand{\hssk}{\hat{\mathbf{s}}_{k}}
\newcommand{\fssk}{\mathbf{s}_{k+1}}
\newcommand{\pssk}{\mathbf{s}_{k-1}}
\newcommand{\psskm}{\mathbf{s}_{m,k-1}}
\newcommand{\sskmpo}{\mathbf{s}_{m, k+1}}
\newcommand{\vsk}{\mathbf{v}_{k}}

\newcommand{\weightkm}{{w}_{m,k}}
\newcommand{\weightkkm}{{w}_{m,k | k}}

\newcommand{\weightkpokm}{{w}_{m,k+1 | k}}
\newcommand{\weightkkmom}{{w}_{m,k | k -1}}

\newcommand{\weightkmo}{{w}_{m,k-1}}

\newcommand{\xs}{x}
\newcommand{\ys}{y}
\newcommand{\zs}{z}

\newcommand{\xsk}{x_{k}}
\newcommand{\ysk}{y_{k}}

\newcommand{\zsk}{z_{k}}

\newcommand{\vxk}{{v}_{\mathsf{x},k}} 
\newcommand{\vyk}{{v}_{\mathsf{y},k}} 
\newcommand{\vzk}{{v}_{\mathsf{z},k}}

\newcommand{\Pk}{\boldsymbol{\Sigma}_{k}}
\newcommand{\Ak}{\mathbf{A}_{k}}
\newcommand{\Q}{\mathbf{Q}}
\newcommand{\Qk}{\mathbf{Q}_{k}}
\newcommand{\Rk}{\mathbf{R}_{k}}
\newcommand{\tRk}{\tilde{\mathbf{R}}_{k}}
\newcommand{\tRmk}{\tilde{\mathbf{R}}_{m,k}}
\newcommand{\R}{\mathbf{R}}

\newcommand{\gammat}{\gamma_{\mathsf{t}}}
\newcommand{\gammam}{\gamma_{\mathsf{m}}}
\newcommand{\Daa}{\mathbf{D}_{k-1}^{11}}
\newcommand{\Dab}{\mathbf{D}_{k-1}^{12}}
\newcommand{\Dba}{\mathbf{D}_{k-1}^{21}}
\newcommand{\Dbb}{\mathbf{D}_{k-1}^{22}}

\newcommand{\Nmc}{N_{\mathsf{mc}}}

\newcommand{\A}{\mathbf{A}}

\newcommand{\bmumk}{\bm{\mu}_{m,k}}
\newcommand{\Smk}{\bm{\Sigma}_{m,k}}
\newcommand{\betak}{\boldsymbol{\eta}_{k}}
\newcommand{\tbetak}{\tilde{\boldsymbol{\eta}}_{k}}
\newcommand{\Kk}{\mathbf{K}_{k}}
\newcommand{\Sk}{\mathbf{S}_{k}}
\newcommand{\Pkk}{\mathbf{P}_{k\rvert k}}
\newcommand{\vk}{\mathbf{v}_{k}}
\newcommand{\zk}{\mathbf{z}_{k}}
\newcommand{\fzk}{\mathbf{z}_{k+1}}
\newcommand{\zotok}{\mathbf{z}_{1:k}}

\newcommand{\zztok}{\mathbf{z}_{0:k}}
\newcommand{\bJk}{\mathbf{J}_k}
\newcommand{\pbJk}{\mathbf{J}_{k-1}}
\newcommand{\PCRB}{\mathbf{P}_{0:k}}

\newcommand{\bJztok}{\mathbf{J}_{0:k}}

\newcommand{\bJkmo}{\mathbf{J}_{k-1}}
\newcommand{\bJz}{\mathbf{J}_0}
\newcommand{\bmk}{\mathbf{m}_{k}}
\newcommand{\bmkk}{\mathbf{m}_{k \lvert k}}
\newcommand{\bmkkmo}{\mathbf{m}_{k \lvert k-1}}
\newcommand{\bmkpok}{\mathbf{m}_{k+1 \lvert k}}

\newcommand{\bmo}{\mathbf{m}_{0}}

\newcommand{\bpo}{\mathbf{q}_0}

\newcommand{\bsso}{\mathbf{s}_0}
\newcommand{\bwk}{\mathbf{w}_{k}}
\newcommand{\bzk}{\mathbf{z}_{k}}

\newcommand{\dF}{d_\mathrm{F}}

\newcommand{\gnk}{g_{n,k}}

\newcommand{\hbsml}{\hat{\mathbf{s}}_{\mathsf{ML},k}}
\newcommand{\BSml}{{\bm{\Sigma}_{\mathbf{s},k}}}

\newcommand{\Hk}{\mathbf{H}_{k}}
\newcommand{\Hkkmo}{\mathbf{H}_{k \lvert k-1}}
\newcommand{\Hmkkmo}{\mathbf{H}_{m,k \lvert k-1}}

\newcommand{\IN}{\mathbf{I}_N}

\newcommand{\Ny}{N_{\mathsf{y}}}
\newcommand{\Nz}{N_{\mathsf{z}}}

\newcommand{\Pkkmo}{\mathbf{P}_{k \lvert k-1}}
\newcommand{\Pkpok}{\mathbf{P}_{k+1 \lvert k}}

\newcommand{\po}{p_0}
\newcommand{\Po}{\boldsymbol{\Sigma}_{0}}
\newcommand{\ps}{{\mathbf{p}}}

\newcommand{\psk}{\mathbf{p}_{k}}

\newcommand{\sigmaketa}{\sigma_{\eta,k}}
\newcommand{\sigmaeta}{\sigma_{\eta}}
\newcommand{\ssigmaeta}{\sigma^2_{\eta}}

\newcommand{\thetank}{\vartheta_{n,k}}

\newcommand{\tilden}{{\tilde n}}

\newcommand{\tra}{\mathsf{T}}

\newcommand{\znk}{z_{n,k}}

\newcommand{\sumnynz}{\sum_{\ny, \nz}}
\newcommand{\sumny}{\sum_{\ny=0}^{\Ny-1}}
\newcommand{\sumnz}{\sum_{\nz=0}^{\Nz-1}}

\newcommand{\param}{\xi_k}

\title{Near-field Tracking with Large Antenna Arrays:\\ Fundamental Limits and Practical Algorithms}
\author{Anna Guerra,~\IEEEmembership{Member,~IEEE,} Francesco Guidi,~\IEEEmembership{Member,~IEEE,} \\ Davide Dardari,~\IEEEmembership{Senior,~IEEE,} and Petar M. Djuri\'c,~\IEEEmembership{Fellow,~IEEE,}

\thanks{A. Guerra (corresponding author, e-mail: anna.guerra3@unibo.it) and D. Dardari  are with the WiLAB - Department of Electrical and Information Engineering ``Guglielmo Marconi" - CNIT, University of Bologna, Italy. F. Guidi is with CNR-IEIIT, Italy. P. M. Djuri\'c is with ECE, Stony Brook University, Stony Brook, NY 11794, USA.
E-mail: petar.djuric@stonybrook.edu.
}
}
\theoremstyle{plain}
\newtheorem{prop}{Proposition}
\newtheorem{assumption}{Assumption}

\theoremstyle{remark}
\newtheorem{rem}{Remark}
\begin{document}
\maketitle
\begin{abstract}
Applications towards 6G have brought a huge interest towards arrays with a high number of antennas and operating within the millimeter and sub-THz bandwidths for joint communication and localization.
With such large arrays, the plane wave approximation is often not accurate because the system may operate in the near-field propagation region (Fresnel region) where the electromagnetic field wavefront is spherical. 
In this case, the \ac{CoA} is a measure  of the spherical wavefront that can be used to infer the source position using only a  single large array.
In this paper, we study a near-field tracking problem for inferring the state (i.e., the position and velocity) of a moving source with an ad-hoc observation model that accounts for the phase profile of  
a large receiving array. For this tracking  problem, we derive the \ac{P-CRLB} and show the effects 
when the source moves inside and outside the Fresnel region. We provide insights on how the loss of positioning information outside  Fresnel comes from  an increase of the ranging error rather than from inaccuracies of angular estimation. Then, we investigate the  performance of different Bayesian tracking algorithms in the presence of model mismatches and abrupt trajectory changes. Our results demonstrate the feasibility and high accuracy for most of the tracking approaches without the need of wideband signals and of any synchronization scheme.
\end{abstract}

\acresetall

\bstctlcite{IEEEexample:BSTcontrol}

\IEEEpeerreviewmaketitle

\section{Introduction}
\label{sec:intro}
Short-range localization and tracking techniques have recently attracted great interest in all the scenarios where the signal coming from the \ac{GNSS} is denied or leads to a low-accuracy positioning \cite{WinSheDai:J18,DarCloDju:J15,GenEtAl:J14}. 
Nowadays, there is a large variety of ad-hoc solutions for high-accuracy positioning, spanning from systems based on dedicated impulse radio \ac{UWB} technology 
to system integrating heterogeneous sensors \cite{win2011network}.
Unfortunately, most of the available solutions usually require the deployment of an {\em ad-hoc} positioning infrastructure with multiple anchors, i.e., multiple reference sensors with known positions, that can be expensive or bulky, especially in indoor environments. 
While it is possible, in principle, to avoid the need of  an infrastructure using \ac{SLAM} algorithms based on laser or camera sensors \cite{durrant2006simultaneous}, it is of interest to realize high-accuracy radio localization and tracking solutions that make use of the same network of \acp{AP} already deployed for communication coverage.   
\begin{figure}[t!]
\centering
\includegraphics[width=0.9\linewidth,draft=false]
{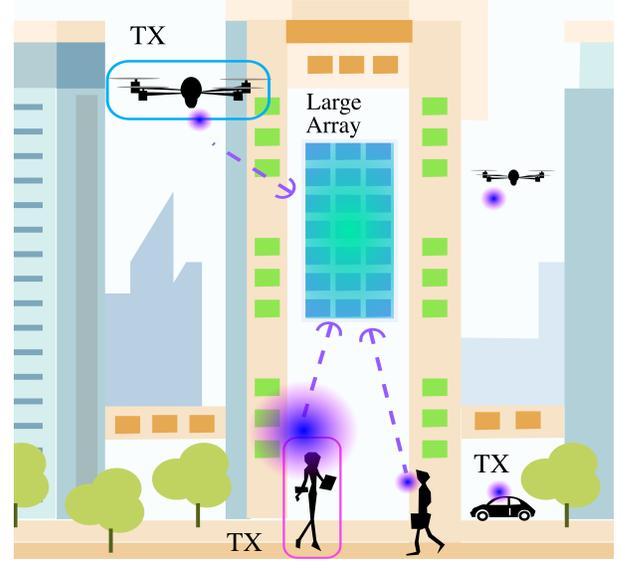}
\caption{Localization scenario with a large array coating a building wall.
}
\label{fig:scenario}
\end{figure}

With the \ac{6G} cellular networks, further improvements are expected in localization and tracking. The improvements will result from the joint use of high frequencies and large arrays for both communication and localization purposes \cite{guidi2019radio, de2021convergent,FasEtAl:J19,GuiEtAl:J18,vukmirovic2018position,GueGuiDar:J18} (see Fig.~\ref{fig:scenario}). Following a trend started by the {5G} cellular systems, a huge number of \acp{AP}, equipped with massive arrays, are expected to play a dual functional role of communication and localization reference nodes. The large arrays at each \ac{AP} allow to collect a huge number of measurements, thus enhancing the localization accuracy.

Usually, with such large arrays, localization is based on the joint estimate of the \ac{AOA} and \ac{TOA} \cite{garcia2017direct,WanWuShe:J19,Wym:C20}, which requires a fine synchronization between the transmitter (namely, the source to be localized) and the receiver (the \ac{AP}). When the synchronization is not guaranteed, it is not possible to retrieve any reliable positioning information about the transmitter if only one \ac{AP} is involved in the process. Traditionally, \ac{TDOA} or two-way ranging approaches are used to overcome this issue \cite{DarCloDju:J15}, but they require multiple message-passing between the two nodes or the involvement of multiple \acp{AP} with  a good \ac{GDOP}. When \acp{AP} are  closely located to each other and latency requirements become stringent, these approaches could fail and, therefore, new solutions are needed.  

When the antenna array is large enough to capture the spherical characteristic of the incident wave, which happens when operating in the radiating near-field of the array (Fresnel region), a promising approach is to retrieve the source position directly from the \ac{CoA} encapsulated in the spherical wavefront impinging a single large array. The \ac{CoA} depends on the transmitter position and the array geometry, and, when it is used for localization purposes, it does not need any synchronization  \cite{ZhaEtAl:J19,guidi2019radio,elzanaty2020reconfigurable}.
This concept is not new and it has been investigated for different frequencies and architectures \cite{HadSacFas:C17,ElKetAl:J13,LeC:J95,friedlander2019localization}, entailing the adoption of distributed antennas \cite{vukmirovic2019direct, ZhaEtAl:J19}.
In \cite{HadSacFas:C17}, an approach for direct wireless positioning with narrowband signals with multi-tone signalling and multi-arrays is described, whereas 
in \cite{ElKetAl:C09-2} a MUSIC-based method and an extensive analysis on the attainable fundamental localization limits is derived for near-field propagation conditions. 
A detailed investigation using acoustic waveforms has been carried out in \cite{LeC:J95,FerWyb:C01}. Unfortunately, these studies usually refer to acoustic waves or \ac{RF} microwave  considering only very short distances or using very large, often not practical, antennas.
With the introduction of the \ac{mm-wave} technology, source positioning and tracking is in principle possible even with antenna arrays with limited aperture and for distances of several meters \cite{guidi2019radio,ZhaEtAl:J19}. Preliminary studies on near-field fundamental limits on positioning with \ac{5G} antenna arrays has been recently addressed in  \cite{HuRusEdf:J18,Wym:C20,bjornson2021reconfigurable}, but considering a static scenario and non-Bayesian methods.

In this paper, we investigate the  fundamental limits in source tracking  in a single array scenario, and we assess the performance through practical algorithms working with \ac{CoA}.  
To this end, we consider an \emph{ad-hoc} phase-based observation model accounting for the near-field wavefronts, and we derive compact formulas for different array configurations to gain further insights on  the capability to infer the position information when moving from near-field to far--field regions, conventionally delimited by the Fraunhofer distance \cite{Bal:05}. Through an asymptotic analysis, we evaluate the role of ranging and bearing information on localization when the source-array distance increases, showing that the \ac{CoA} provides both types of information only in the Fresnel region while, elsewhere, only bearing data can be correctly estimated. Further, we investigate different Bayesian tracking algorithms to assess their robustness and accuracy in different situations.

The main contributions of the paper are as follows.
\begin{itemize}
    \item We introduce a narrowband observation model, accounting for phase difference-of-arrival at a single large array, that includes \ac{CoA}  of the impinging wavefront;
     \item We derive the \ac{P-CRLB} to assess the ultimate performance of the \ac{CoA}--based tracking in the near-- and far--field regions when the considered phase-based model is employed; 
     \item We derive compact formulas for the \ac{FIM} on ranging and bearing information for two different array geometries and we highlight the role of the ratio between the array aperture and the source distance in defining the near-field localization coverage;
    \item We evaluate the performance of different Bayesian filtering approaches considering different parameter models available at the receiver. We investigate the robustness of the  tracking algorithms with respect to model parameter mismatches, abrupt changes of direction, and the impact of movements inside/outside the Fresnel region.
\end{itemize}

\paragraph{Notation} Scalar variables, vectors and matrices are represented with lower letters, lower bold letters, and capital bold letters, respectively (e.g., $x$, $\mathbf{x}$, and $\mathbf{X}$, respectively). The symbols $\left( \cdot \right)^\tra$, $\left( \cdot \right)^{-1}$, and $\left( \cdot \right)^\dagger$ represent the transpose, inverse and Moore-Penrose pseudo-inverse operators of their arguments, respectively, and $\lVert \cdot \rVert$ is the 2-norm of its argument. We use $k$ for discrete temporal indexing, $n$ for antenna indexing, and $m$ for particle indexing. As an example, $x_{n,k}$, $\mathbf{x}_{n,k}$, $\mathbf{X}_{n,k}$ stand for a scalar, a vector or a matrix related to the $n$th antenna at the $k$th time instant. With $\IN$ and $\mathbf{0}_N$ we represent the identity and all-zero matrices of size $N \times N$, \and with $p\left( \cdot \right)$ probability density functions (pdfs). 
With $\mathcal{N}\left(\mathbf{x}; \bm{\mu},  \bm{\Sigma} \right)$ we indicate that the random vector $\mathbf{x}$ is distributed according to a Gaussian pdf with a mean vector $\bm{\mu}$ and a covariance matrix $\bm{\Sigma}$. The notation $\mathbf{x}_{a \lvert b}$ indicates the value of a vector $\mathbf{x}$ at time instant $a$ estimated by considering the measurements collected up to time instant $b$. For example, $\mathbf{x}_{k \lvert k-1}$ is the value of $\mathbf{x}$ predicted at time instant $k-1$ for the next time instant $k$, whereas, once a new measurement becomes available at $k$, this value is updated to  $\mathbf{x}_{k \lvert k}$.

\paragraph{Organisation of the paper} 
The rest of the paper is organized as follows. Section~\ref{sec:signmodel} provides the state-space model of the tracking problem, whereas Sections~\ref{sec:tracklim}-\ref{sec:track} describe the  fundamental limits of localization performance and present practical algorithms for source tracking in near-field, respectively.  
A case study is addressed  in Section~\ref{sec:results} and conclusions are drawn in Section~\ref{sec:conclusions}.
%
%
\section{State--Space Model}\label{sec:signmodel}
We consider a tracking scenario where a single antenna array tracks a moving source by exploiting the phase profile of the received signal caused by the \ac{CoA}. We denote by $\ssk=\left[ \psk^{\tra},\, \vsk^{\tra}  \right]^{\tra}$ the state composed of the position and velocity Cartesian coordinates of the source at time instant $k$, respectively, defined by
$\psk= \left[ \xsk,\, \ysk, \zsk \right]^{\tra}$ and 
$\vsk=\left[ \vxk,\, \vyk, \vzk \right]^{\tra}$.
Therefore, the state dimensionality is $\Ns=6$, because of the $3$D position and velocity Cartesian coordinates.
We also consider that the array has $N$ antennas located at {$\pnk=\left[ \xn,\, \yn, \zn \right]^{\tra}$}, $n=0, \ldots, (N-1)$, with reference location $\pok$. 

At each time instant, the geometric relationship between the reference location and the source is given by
 \begin{align}\label{eq:cartesian}
    \!\! \psk &\!=\! \left[ \begin{array}{l}
     \xsk \\
     \ysk \\
     \zsk
     \end{array}
     \right]\!=\! \!\left[ \begin{array}{l}
     \!\xo+ \dk\, \cos\left( \phik \right) \,\sin\left( \thetak \right) \\
     \yo+\dk\, \sin\left( \phik \right) \sin\left( \thetak \right) \\
     \zo+ \dk\, \cos\left( \thetak \right)
     \end{array}
     \right],
 \end{align}
with $\dk=\lVert \psk - \pok\rVert$, $\phik=\operatorname{atan2}\left( \frac{\ysk-\yo}{\xsk-\xo}\right)$,  and $\thetak=\operatorname{acos}\left( \frac{\zsk-\zo}{\dk}\right)$  being the true distance, azimuth, and elevation angles, respectively, 
as represented in Fig.~\ref{fig:geo}.

The source emits a narrowband signal such as a tone or a pilot in a resource block of an OFDM scheme with a frequency $\fc$, which is received by the antenna array and processed for tracking purposes. We assume that the source is not synchronized with the receiver so that any \ac{TOA} information cannot be inferred from the signal and the phase offset between the source and the array is not known. Starting from the collected phase measurements at each antenna of the array, the purpose is to estimate and track the state of the source.
\begin{figure}[t!]
\centering
\psfrag{x}[lc][lc][0.7]{$x$}
\psfrag{y}[lc][lc][0.7]{$y$}
\psfrag{z}[lc][lc][0.7]{\,\,$z$}
\psfrag{Ant}[lc][lc][0.7]{Large array}
\psfrag{da}[c][c][0.7]{$\dno$ \,\,}
\psfrag{ta}[c][c][0.7]{$\thetano$}
\psfrag{pa}[c][c][0.7]{\quad$\phino$}
\psfrag{p}[c][c][0.7]{\quad$\phik$}
\psfrag{t}[c][c][0.7]{\quad$\thetak$}
\psfrag{d}[c][c][0.7]{\quad$\dk$}
\psfrag{r}[c][c][0.7]{$\pok$}
\psfrag{a}[c][c][0.7]{$\pnk$}
\psfrag{q}[lc][lc][0.7]{Trajectory}
\psfrag{s}[c][c][0.7]{$\psk$}
\includegraphics[width=0.9\linewidth,draft=false]
{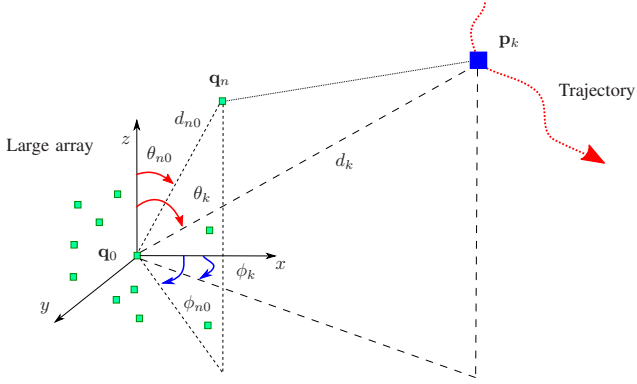}
\caption{Considered array geometry. The receiver array (antennas indicated with green squares with reference location in $\pok$), tracks a moving source at $\psk$ (blue square marker) whose trajectory is depicted with a dashed red line. 
}
\label{fig:geo}
\end{figure}

The sequential state estimation problem (tracking)  can be formulated starting from a discrete-time state-space representation given by \cite{sarkka2013bayesian}
\begin{align} \label{eq:motionmodel}
&\ssk = f\left(\pssk  \right) + \bwk = \Ak \, \pssk  + \bwk ,  \\
& \bzk = h\left(\psk \right) + \betak, 
\label{eq:observationmodel}
\end{align}
where the motion model $f: \mathbb{R}^6 \rightarrow \mathbb{R}^{6}$ is considered a linear function of the state, with $\Ak$ being the transition matrix, whereas the observation model is a nonlinear function $h: \mathbb{R}^3 \rightarrow \mathbb{R}^{N}$ that will be defined in the sequel, and $\bwk \sim \mathcal{N} \left(\bwk; \bf 0, \Qk \right)\,$ and $\betak \sim \mathcal{N} \left( \betak; \bf 0, \Rk \right)\,$ are  zero-mean noise processes with $\Qk$ and $\Rk$ being the transition and observation noise covariance matrices. In the next, we will assume a time invariant transition matrix and covariance matrix, e.g., $\Ak=\A$ and $\Qk=\Q$, as well as $\Rk=\R=\sigmaeta^2\, \IN$.

The  observation function provides, for a given source position,  the differential phases at each antenna, i.e., the difference of phases gathered at the considered antenna and at the reference location. More specifically,
\begin{align}\label{eq:observationmodel_v4}
&h\left(\psk \right)=\left[\ho\left(\psk \right),\ldots, \hn\left(\psk \right), \ldots, \hN\left(\psk \right) \right]^{\tra}, 
\end{align}
where the generic element is a phase difference between $0$ and $2\pi$, given by\footnote{Note that the phase uncertainty due to source-array clock mismatches disappears thanks to the operation of difference between the phases at the array antennas and the reference.}
\begin{align}\label{eq:observationmodel_v2}
&\hn\left(\psk \right)=\Delta \thetank \,\, \operatorname{mod}\,\, 2\pi, \\
&\Delta \thetank = \frac{2\,\pi}{\lambdap} \, \Delta \dnk\left(\pnk,\psk \right), \label{eq:observationmodel_v3}
\end{align}
\noindent where $\Delta \thetank$ represents the phase difference between locations $\psk$ and the reference location $\pok$,  $\operatorname{mod}$ is the modulo operator that returns the remainder after division (of  $\Delta \thetank/ 2\pi$) with the same sign of $\Delta\thetank$, $\lambdap$ 
is the wavelength,
 and $\Delta \dnk \left(\pnk,\psk \right)$ is the extra-distance traveled by the waveform to arrive to the $n$th antenna with respect to the reference one. In particular, this extra-distance is given by
\begin{align}\label{eq:deltad}
\Delta \dnk \left(\pnk,\psk \right)= \dnk\left(\pnk,\psk \right)-\dk\left(\pok,\psk \right) \,,
\end{align}
with $\dnk$ representing the distance between the $n$th antenna and the source at the time instant $k$. According to trigonometric rules, we have 
\begin{align}
  &\dnk^2 =\dno^2+\dk^2-2\, \dk \,  \dno \, \gnk,
\end{align}
where $\dno=\lVert \pnk -\pok \rVert$ is the distance between the $n$th antenna and the reference location,  and
\begin{align}\label{eq:gn}
  \gnk \triangleq g\left(\thetano, \phino, \thetak, \phik \right)=& \sin\left(  \thetano \right)  \sin\left( \thetak \right)   \cos\left( \phino - \phik \right)  \nonumber \\
  &+ \cos\left( \thetano\right)  \cos\left( \thetak\right),
\end{align}
\begin{figure}[t!]
\centering
\psfrag{y}[lc][lc][0.8]{$y$ [m]}
\psfrag{wrap}[c][c][0.8]{}
\psfrag{unwrap}[c][c][0.8]{}
\psfrag{dphase}[c][c][0.8]{$h_n\left(y \right)$}
\centerline{
\includegraphics[width=0.85\linewidth,draft=false]
{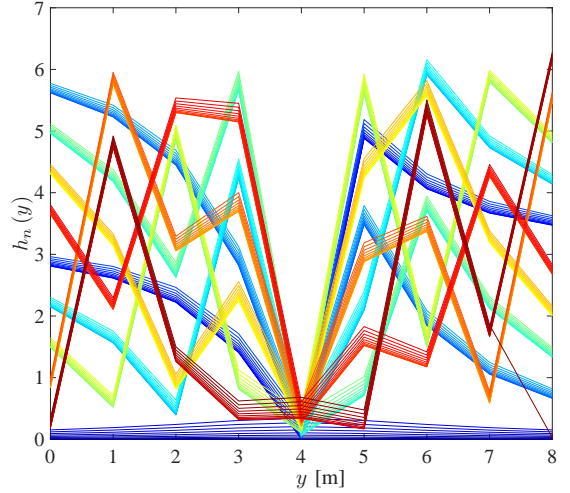}
}
\caption{Differential phases collected by the antennas of an array located at $\bpo=\left[0, 4, 1 \right]$ with $\Ny=\Nz=10$ antennas along the $y$- and $z$- axis, respectively. The target is located at $\left[\xs, \ys, \zs \right]=\left[\xo, y, \zo\right]$ with the $y$ coordinate spanning from $0$ to $8$ m. The central frequency is set to $28$ GHz. The colors of the lines indicate the different antennas of the array, clusters of close antennas measure the same differential phases. }
\label{fig:phase_ex}
\end{figure}
\noindent is a geometric term 
with $\thetano$
and $\phino$
being the $n$th antenna elevation and azimuth angles with respect to the reference location, respectively.
Consequently, the extra-distance in \eqref{eq:deltad} can be written as
\begin{align}\label{eq:a_new}
  \Delta \dnk \left(\pnk,\psk \right) &=  \dk\,\left[ \sqrt{ \fdtk } - 1 \right]\,,
\end{align}
with the \ac{CoA} information gathered in $\fdtk$ as
\begin{align}\label{eq:f}
\fdtk \triangleq f_{n,k}\left(\pn, \psk \right)= 1+\left(\frac{\dno}{\dk}\right)^2-2  \frac{ \dno}{\dk} \,\gnk.
\end{align}
Note that the observation function in \eqref{eq:observationmodel_v2} is highly nonlinear with respect to the state as highlighted in \eqref{eq:a_new}-\eqref{eq:f} and in the example reported in Fig.~\ref{fig:phase_ex} where the differential phase profile is reported as a function of the source's $y$ coordinate. 

In the next sections, we will investigate the theoretical limits on source tracking as well as some practical algorithms by considering the source located both in the radiating near-- and far--field regions. Conventionally, the far--field region corresponds to distances larger than the Fraunhofer limit given by $\dF \triangleq \frac{2\,D^2}{\lambdap}$, where $D$ is the diameter of the antenna array, whereas the radiating near--field region is \cite{Bal:05}
\begin{align}
& 0.62\, \sqrt{\frac{D^3}{\lambdap}}\le \dk \le \dF.
\end{align}
\section{Fundamental Limits on Near--field Tracking}\label{sec:tracklim}
\subsection{Posterior CRLB}
In this section, we derive the \ac{P-CRLB} \cite{tichavsky1998posterior,brehard2006closed,bergman2001optimal,fritsche2014fresh} for the discrete-time nonlinear problem described in this paper. As introduced in \cite{tichavsky1998posterior,fritsche2014fresh}, different Bayesian bounds can be derived depending on the choice of the probability distribution from which the Bayesian \ac{FIM} is computed. 

The joint distribution of the state and measurements, $p_{0:k}=p\left( \ssztok, \zztok \right)$, allows for the computation of the Bayesian \ac{FIM} from the state history, i.e., $\bJztok$, and the derivation of the tightest \ac{P-CRLB}  for nonlinear filtering problems. 
The \ac{P-CRLB} of the joint distribution can be written as the inverse of the Bayesian \ac{FIM}
\begin{align}\label{eq:PCRB}
\PCRB \ge \bJztok^{-1},
\end{align}
where $\bJztok$ is the Bayesian \ac{FIM} defined as
\begin{align}\label{eq:jointFIM}
    &\bJztok=\mathbb{E}_{\ssztok, \zztok} \left\{ \Delta_{\ssztok}^{\ssztok} \ln p\left( \ssztok, \zztok \right) \right\}, 
\end{align}
with $\Delta_{\mathbf{a}}^{\mathbf{b}}=\nabla_{\mathbf{a}}^{\tra} \nabla_{\mathbf{b}}$ and $\nabla_{\mathbf{a}}=\left[\frac{\partial}{\partial {a}_1}, \,\ldots \frac{\partial}{\partial {a}_i},\, \ldots,\,  \frac{\partial}{\partial {a}_{\lvert \mathbf{a} \rvert}} \right]$ being the gradient with respect to the vector $\mathbf{a}$. From \eqref{eq:jointFIM}, it is possible to derive the \ac{P-CRLB} on $\ssk$ by
taking the $\Ns \times \Ns$ lower-right sub-matrix of the inverse of $\bJztok$. A more elegant approach that avoids the inversion of the large \ac{FIM} $\bJztok$ is a recursive formula proposed in \cite{tichavsky1998posterior}, which  permits to express
the \ac{FIM} of $\ssk$ as
\begin{align}
 \bJk&=\Dbb  - \Dba \left(\bJkmo + \Daa \right)^{-1}\, \Dab \,,
\end{align}
where the initial information matrix $\bJz=\mathbb{E}_{\bsso} \left\{ \Delta_{\bsso}^{\bsso} \ln p\left( \bsso \right) \right\}$ is derived from the prior distribution of the state \cite{koohifar2018autonomous} and where 
\begin{align}
    &\Daa=\mathbb{E}_{p_{0:k}}\left[- \Delta_{\pssk}^{\pssk} \ln p\left( \ssk \lvert \pssk\right) \right] =\A^{\tra}  \Q^{-1} \A, \\
   & \Dab\!=\mathbb{E}_{p_{0:k}}\left[- \Delta_{\pssk}^{\ssk} \ln p\left( \ssk \lvert \pssk\right) \right]=- \A^{\tra} \Q^{-1}, \\
    &\Dba=\left( \Dab \right)^{\tra}, \\
    &\Dbb=\mathbb{E}_{p_{0:k}}\left[- \Delta_{\ssk}^{\ssk} \ln p\left( \ssk \lvert \pssk\right) - \Delta_{\ssk}^{\ssk} \ln p\left( \zk \lvert \ssk\right) \right]
    \nonumber \\
    &\qquad = \Q^{-1}+ \JM_{k} \,,
\end{align}
where $\JM_{k}$ is the expectation of the Hessian matrix with respect to the state and measurements as \cite{koohifar2018autonomous} 
\begin{align}\label{eq:hessian}
&\JM_{k}=\mathbb{E}_{p_{0:k}}\left\{- \Delta_{\ssk}^{\ssk} \ln p\left( \zk \lvert \ssk\right)\right\} \nonumber \\
&=\mathbb{E}_{\ssk \lvert \pssk}\left\{\mathbb{E}_{\zk \lvert \ssk}\left\{- \Delta_{\ssk}^{\ssk} \ln p\left( \zk \lvert \ssk\right)\right\}\right\} \nonumber \\
&=\mathbb{E}_{\ssk  \lvert \pssk}\left\{{\tJM}_{k}\right\},
\end{align}
with ${\tJM}_{k}$ being the non-Bayesian data \ac{FIM}.
Unfortunately, the expectation in \eqref{eq:hessian} cannot be easily derived, but it is often approximated using Monte Carlo integration \cite{koohifar2018autonomous}. Indeed, we can separate the contribution deriving from the collected data and the prior information, thus writing
\begin{align}\label{eq:FIM_P_M}
    \bJk=\JP_{k \lvert k-1}+\JM_{k}, 
\end{align}
where $\JP_{k \lvert k-1}=\Q^{-1}-\Dba \left(\pbJk + \Daa \right)^{-1}\, \Dab$ contains the propagation information. 
In the sequel, we will find closed-form solutions for ${\tJM}_{k}$ to better investigate the behaviour in near-- and far--field regions. 
In the case study of Section~\ref{sec:results}, we will make use of the Bayesian \ac{FIM} in \eqref{eq:FIM_P_M} as a benchmark for practical tracking algorithms.

\subsection{A Near-- vs. Far--field Fisher Information Analysis} 
\label{subsec:dataFIM}
We now investigate the behavior of the non-Bayesian data \ac{FIM} when the target approaches the Fraunhofer distance, that is, when $\dk = d_{\mathsf{F}}$, and when it is far from the radiating near--field region, i.e., $\dk \gg \dF$.
\begin{prop}
Under the observation model in \eqref{eq:observationmodel_v2}-\eqref{eq:observationmodel_v3} and {weak regularity conditions for $ p\left( \zk \lvert \ssk\right)$ \cite{kay1993fundamentals}}, the positioning information carried by the data \ac{FIM}  vanishes when the distance increases, i.e., 
\begin{align} \label{eq:asyJM}
  {\tJM}_{k} &\triangleq \mathbb{E}_{\zk \lvert \ssk}\left\{\nabla_{\ssk}^{\tra} \ln p\left( \zk \lvert \ssk\right)\, \nabla_{\ssk} \ln p\left( \zk \lvert \ssk\right) \right\} \nonumber \\
  &=\frac{1}{\ssigmaeta}\, \nabla_{\ssk}^{\tra}\, h \left(\psk \right) \nabla_{\ssk}\, h \left(\psk \right) \underset{\dk \gg \dF}{\longrightarrow} \mathbf{0}_{\Ns},
\end{align}
%
where $\mathbf{0}_{\Ns}$ is an all-zero matrix. 
\end{prop}
\begin{proof}
All the entries of the data \ac{FIM}\footnote{{In order to meet the regurality conditions, e.g., to let $\nabla_{\ssk} \ln p\left( \zk \lvert \ssk\right)$ exist and be finite, the derivatives of $h_n(\pk)$ with respect to $\ssk$ are taken equal to the left and right derivatives in the discontinuous point, which is equivalent to substituting $h_n(\pk)$ with $\Delta \vartheta_{n,k}$ in the derivative process.} 
}
\begin{align} 
  \left\{\tJM_{k}\right\}_{i,j} &\!\!\!\!= \frac{1}{\ssigmaeta}\, \left(\frac{2\, \pi}{\lambdap}\right)^2  \sumn \frac{\partial \Delta \dnk }{\partial \left[\ssk\right]_j} \cdot \frac{\partial \Delta \dnk }{\partial \left[\ssk\right]_i} \underset{\dk \gg \dF}{\longrightarrow} 0 \,,
\end{align}
$\forall i,j=1, \ldots, \Ns$, tend to be zero since
\begin{align}\label{eq:limite}
&\frac{\partial \Delta \dnk }{\partial \left[\ssk\right]_i} \underset{\dk \gg \dF}{\longrightarrow} 0, && \forall i=1, \ldots, \Ns.
\end{align}
as demonstrated in Appendix A. \end{proof}
In the next, we will show that, when moving toward the far--field region, such data information vanishing is only caused by the loss of the distance information, whereas the angles can still be inferred. 
To analyze this point, we derive the single components of the \ac{FIM} on the distance and angle parameters, namely $\param \in \left\{\dk, \thetak, \phik \right\}$. We have 
\begin{align}\label{eq:dfim}
\tjM_k\left(\param \right) &\triangleq\mathbb{E}_{\bzk | \param}\left\{ \left(\frac{\partial\ln p\left(\bzk|\param\right) }{\partial \param} \right)^2   \right\}, 
\end{align}
where the gradient of the log-likelihood is given by
\begin{align}
  &\frac{\partial\ln p\left(\bzk|\param\right) }{\partial \param}  
   =   \frac{1}{\sigmaeta^2}\sumn  \frac{\partial h_n\left(\param\right) }{\partial \param} \, (\znk - h_n\left(\param\right)). 
\end{align}
{In the next, for notation simplicity, we omit the time index $k$.}
\begin{prop}\label{prop:p0}
Under the observation model in \eqref{eq:observationmodel_v2}-\eqref{eq:observationmodel_v3} and weak regularity conditions for $ p\left( \zk \lvert \param\right)$, the \acp{FIM} for distance and angles for both near-- and far--field regions (for any value of $d$) and for any geometry (i.e., for any $\dno$, $\gn$) are 
%
\begin{align} \label{eq:Jkdfinal}
&\tjM \left(d \right)= \frac{4\, \pi^2}{\lambdap^2\,\sigmaeta^2} \sumn \frac{1}{d^2+\dno^2-2\,\gn\, \dno\, d}  \nonumber \\
&\phantom{\tjM \left(d \right)=}\times \Big[ 2\,d^2 + \dno^2\,(\gn^2+1) -4\, g_n\, d\, \dno  +  \nonumber \\
&\phantom{\tjM \left(d \right)=} -2 (d- \gn \dno) \sqrt{d^2+\dno^2-2\, \gn d\, \dno} \Big], \\
&\tjM\left(\theta \right)  = \frac{4\,\pi^2}{\lambdap^2 \sigmaeta^2}\sumn\frac{d^2\,\dno^2}{\dno^2+ d^2-2\,\gn\,d\,\dno} \left(\frac{\partial \gn}{\partial \theta} \right)^2, \\
&\tjM \left(\phi \right)  = \frac{4\,\pi^2}{\lambdap^2 \, \sigmaeta^2}\sumn\frac{d^2\,\dno^2\,}{{\dno^2}+ d^2-2\, {\gn\,d\,\dno}}\left(\frac{\partial \gn}{\partial \phi} \right)^2.
  \label{eq:Jkthetafinal}
\end{align}  
\end{prop}
\begin{proof}
See Appendix B.
\end{proof}
Because \eqref{eq:Jkdfinal}-\eqref{eq:Jkthetafinal} are not easy to interp, we further simplify by focusing on planar circular and rectangular array geometries, and hence,  we assume
\begin{assumption}[Planar Circular Array]\label{as:A1} The distance and azimuth angle between the $n$th antenna and the reference location are $\dno=D/2,\,\phino =\frac{\pi}{2},\,\forall \, n$, whereas the elevation angle is set to  $\thetano  = \frac{n\, 2\, \pi}{N}$ (lying on $YZ$ plane).
\end{assumption}
\begin{assumption}[Source Position]\label{as:A2} The source is on the central perpendicular line, along the $x$-axis, so that $\theta=\frac{\pi}{2},\phi=0$, such that $g_n=0,\,\forall \, n=0, \ldots, N-1$.
\end{assumption}
\begin{prop}\label{prop:p1}
Under Assumption \ref{as:A1}, the \acp{FIM} in  \eqref{eq:Jkdfinal}-\eqref{eq:Jkthetafinal} for a generic source position are
\begin{align}\label{eq:Jkdfinal1}
&\tjM \left(d \right)\! = \!\frac{4\, \pi^2}{\lambdap^2\,\sigmaeta^2} \sumn \frac{1}{4\, d^2+D^2-4\, g_n d\, D}  \nonumber \\
&\phantom{\tjM \left(d \right)=}\times \Big[ 8\,d^2 + D^2\,(g_n^2+1) -8\, g_n\, d\, D  + \nonumber \\
&\phantom{J \left(d \right)=} -2 (2\,d- g_n D) \sqrt{4\, d^2+D^2-4\, g_n d\, D} \Big],\\
&\tjM \left(\theta \right)  = \frac{4\,\pi^2}{\lambdap^2 \,\sigmaeta^2}\sumn\frac{d^2\,D^2}{{D^2}+ 4\,d^2-4\, {g_n\,d\,D}} \left(\frac{\partial \gn}{\partial \theta} \right)^2, \\
&\tjM \left(\phi \right)  = \frac{4\,\pi^2}{\lambdap^2 \,\sigmaeta^2}\sumn\frac{d^2\,D^2}{{D^2}+ 4\,d^2-4\, {g_n\,d\,D}} \left(\frac{\partial \gn}{\partial \phi} \right)^2.
\label{eq:Jkthetafinal1}
\end{align}
\end{prop}
\begin{proof}
From \eqref{eq:Jkdfinal}-\eqref{eq:Jkthetafinal}, they can be obtained by substituting $\left(\dno, \phino, \thetano \right)$ according to Assumption \ref{as:A1}.
\end{proof}
\begin{prop}\label{cor:c1}
Under Assumptions \ref{as:A1}-\ref{as:A2}, the \acp{FIM} in  \eqref{eq:Jkdfinal}-\eqref{eq:Jkthetafinal} for a circular array on the $YZ$-plane and a target on the $x$-axis are
\begin{align}\label{eq:CFIMd}
&\tjM\left(d \right) 
= \frac{4\, N\,\pi^2}{\lambdap^2\, \sigmaeta^2} \cdot \frac{2\, +\frac{D^2}{4\,d^2}-2\, \sqrt{1+\frac{D^2}{4\,d^2}}}{1+\frac{D^2}{4\,d^2}},  \\
&\tjM\left(\theta \right) = \tjM\left(\phi \right) = 
\frac{N\, \pi^2}{2\, \lambdap^2 \sigmaeta^2}\, \frac{D^2}{1+ \frac{D^2}{4\, d^2}}.
\label{eq:CFIMp}
\end{align}
\end{prop}
\begin{proof}
See Appendix C.
\end{proof}
\begin{figure}[t!]
\psfrag{T}[c][c][0.9]{Circular array}
\psfrag{x}[c][c][0.9]{$d/D$}
\psfrag{y}[c][r][0.9]{$\sqrtJd$}
\psfrag{data11111111111111111}[lc][lc][0.7]{$D=0.14\,$m}
\psfrag{data2}[lc][lc][0.7]{$d=\dF,\,D=0.14\,$m}
\psfrag{data4}[lc][lc][0.7]{$D\!=\!0.75\,$m}
\psfrag{data5}[lc][lc][0.7]{$d=\dF,\,D=0.75\,$m}
\psfrag{data1}[c][c][0.7]{\quad \,\,\, Threshold}
\centerline{
\includegraphics[width=0.85\linewidth,draft=false]
{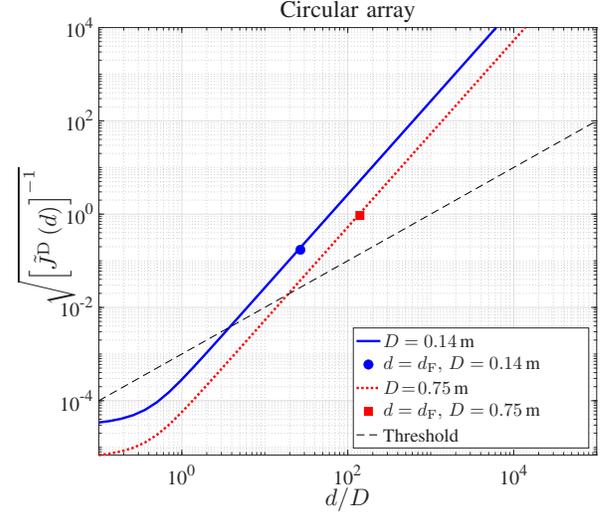}}
\caption{Ranging error for circular arrays and for different array diameters as a function of $d/D$. The central frequency is set to $28\,$GHz. The standard deviation of the measurement noise is fixed to $20^\circ$. The numbers of antennas are $N=400$ (blue continuous curve) and $N=10000$ (red dotted curve). The threshold line indicates a ranging error of $0.1\, \%$ of the actual distance.}
\label{fig:crbd_circ}
\end{figure}
\begin{rem}\label{rem:r1}
From \eqref{eq:CFIMd}-\eqref{eq:CFIMp}, we obtain the \acp{FIM} at the boundary of the Fresnel region ($\dk=\dF$), 
\begin{align}\label{eq:Jcircdf}
&\tjM\left(d \right)\! = \frac{4\, N\,\pi^2}{\lambdap^2\, \sigmaeta^2} \, \frac{2\, +\frac{\lambdap^2}{16\, D^2}-2\, \sqrt{1+\frac{\lambdap^2}{16\, D^2}}}{1+\frac{\lambdap^2}{16\, D^2}}, \\
&\tjM\left(\theta \right)  = \tjM \left(\phi \right)  = \frac{N\, \pi^2}{2\, \lambdap^2 \sigmaeta^2}\, \frac{D^2}{1+ \frac{\lambdap^2}{16\, D^2}}, 
\label{eq:Jcirctf}
\end{align}
and, for $\dk \gg \dF$ (far--field region), we get 
\begin{align}
\tjM\left(d \right)&=0,\\
\label{eq:asym}
\tjM\left(\theta \right)&=\tjM\left(\phi \right)=\frac{D^2\,N\, \pi^2}{2\, \lambdap^2 \sigmaeta^2}.
\end{align}
The results show the dependence of the FIMs on the diameter $D$ of the array and the number of measurements $N$. Further, they also reveal that the FIMs are inversely proportional to the measurement noise variance  $\sigmaeta^2$ and squared wavelength $\lambdap^2$.

Note that $\tjM\left(\theta \right)$ and $\tjM \left(\phi \right)$ in \eqref{eq:Jcirctf} tend to their asymptotic values in \eqref{eq:asym} because  $\frac{\lambda^2}{16\,D^2} \ll 1$. According to this result, 
outside the near--field region bounded by $\dF$, it is not possible to  retrieve the target position because the \ac{CoA} tends to vanish, despite the feasibility of estimating the angles.
Figure~\ref{fig:crbd_circ} displays the square root of the inverse of the ranging \ac{FIM}, $\tjM(d)$, as a function of the source-array distance normalized with respect to the array diameter. The markers indicate the value of \eqref{eq:Jcircdf} at the Fraunhofer distance. We observe from the obtained results that the ranging information depends on the ratio between the source distance $d$ and the array size, represented by its diameter $D$, and tends to decrease when this ratio is large. 
The figure also shows a threshold line that corresponds to $0.1\%$ error of the actual distance. We can see that the 
inverse of the ranging \ac{FIM}
is above the threshold outside the Fraunhofer boundary.
\end{rem}

We now consider a rectangular array lying on the $YZ$-plane with $\pok=\left[0, 0, 0 \right]$ and antennas equally spaced by $\lambdap/2$. The generic $n$th antenna is located at $\pnk= \frac{\lambdap}{2}\,\left[ 0, \ny, \nz \right]$ with $\ny=0, \ldots, \Ny-1$ and $\nz=0, \ldots, \Nz-1$ being the antenna index along the $y$- and $z$-axis, respectively. We, thus, have the following assumption:
\begin{assumption}[Planar Rectangular Array]\label{as:A3} The distance, azimuth and elevation angles between the $n$th antenna and the reference location are $\dno=\frac{\lambdap}{2} \, \sqrt{\ny^2+\nz^2} \triangleq \frac{\lambdap}{2} \, \ntilde$, $\phino= \frac{\pi}{2}$, and $\thetano=\operatorname{acos}\left( \frac{\nz}{\ntilde}\right) \in\left[0, \frac{\pi}{2} \right]$, respectively.
\end{assumption}
\begin{prop}\label{prop:p3}
Under Assumptions \ref{as:A2}-\ref{as:A3}, considering a rectangular array on the $YZ$-plane and a source on the $x$-axis, the \acp{FIM} in \eqref{eq:Jkdfinal}-\eqref{eq:Jkthetafinal} become
\begin{align}\label{eq:Jkdfinal2}
&\tjM\left(d \right)=  \frac{4\,\pi^2}{\lambdap^2\, \sigmaeta^2}\sumny\sumnz \frac{1}{4\, d^2+\lambdap^2\, \ntilde^2} \, \big[ 8\, d^2 + \nonumber \\
&\phantom{J_k \left(\dk \right)=} + \lambdap^2\, \ntilde^2-4\, d\, \sqrt{4\, d^2+\lambdap^2\, \ntilde^2} \big], \\
&\tjM\left(\theta \right)  =  \frac{4\,\pi^2}{ \sigmaeta^2}\sumny \sumnz \frac{\tilden^2\,d^2}{\tilden^2\,\lambdap^2+4\,d^2}\left(\cos \thetano\right)^2,\\
&\tjM \left(\phi \right)  =  \frac{4\,\pi^2}{ \sigmaeta^2}\sumny \sumnz \frac{\tilden^2\,d^2}{\tilden^2\,\lambdap^2+4\,d^2}\left(\sin \thetano\right)^2.
\label{eq:Jkthetaplanar}
\end{align}
\end{prop}

\begin{proof}
From \eqref{eq:Jkdfinal}-\eqref{eq:Jkthetafinal}, the \acp{FIM} can be obtained by substituting $\left(\dno, \phino, \thetano \right)$ according to Assumption \ref{as:A2} and $\left(\theta, \phi \right)$ according to Assumption \ref{as:A3}.
\end{proof}

\begin{rem}\label{rem:r2} For asymptotic considerations, we specialize \eqref{eq:Jkdfinal2}-\eqref{eq:Jkthetaplanar} at the boundary of the Fresnel region ($d=\dF$), considering $D= \frac{\lambdap}{2} \Ntilde \triangleq \frac{\lambdap}{2} \sqrt{\Ny^2+\Nz^2}$, thus yielding 
\begin{align}
\tjM\left(d \right) & 
=\frac{4\, \pi^2}{\lambdap^2\, \sigmaeta^2} \sumnynz\frac{2\, +\left(\frac{\ntilde}{\Ntilde^2}\right)^2-2\, \sqrt{1+\left(\frac{\ntilde}{\Ntilde^2}\right)^2}}{1+\left(\frac{\ntilde}{\Ntilde^2}\right)^2}, \\
\tjM \left(\theta \right)&=\frac{\pi^2}{ \sigmaeta^2}\sumnynz \frac{\ntilde^2}{1+ \left(\frac{\ntilde}{\Ntilde^2}\right)^2}\left(\cos \thetano\right)^2,   \\
\tjM \left(\phi \right)&=\frac{\pi^2}{ \sigmaeta^2}\sumnynz \frac{\ntilde^2}{1+ \left(\frac{\ntilde}{\Ntilde^2}\right)^2}\left(\sin \thetano\right)^2, 
\end{align}
where $\sumnynz=\sumny \sumnz$, whereas, for $\dk \gg \dF$, we obtain 
\begin{align}\label{eq:FIMdinf}
&\tjM\left(d \right) = 0,  \\
&\tjM \left(\theta \right) 
=  \frac{\pi^2}{ \sigmaeta^2}\,\frac{\Ny \, \Nz \, \left(2 \, \Nz - 1 \right)\, \left( \Nz - 1 \right) }{6},  \\
&\tjM \left(\phi \right) 
=\frac{\pi^2}{ \sigmaeta^2}\,\frac{\Ny \, \Nz \left(2 \Ny - 1\right)\,\left(\Ny - 1\right)}{6}\,. \label{eq:FIMpinf}
\end{align}
Notably, due to the considered system geometry, the number of antennas on the $z$-axis, i.e., $\Nz$, augments the information in estimating the elevation angle, whereas $\Ny$ plays the same role for the azimuth. 
\begin{figure}[t!]
\psfrag{T}[c][c][0.9]{Rectangular array}
\psfrag{x}[c][c][0.9]{$d/D$}
\psfrag{y}[c][r][0.9]{$\sqrtJd$}
\psfrag{data11111111111111111111}[lc][lc][0.7]{$D=0.14\,\text{m},\,N=20\times 20$}
\psfrag{data2}[lc][lc][0.7]{$d=\dF,\,D=0.14\,$m}
\psfrag{data7}[lc][lc][0.7]{$D\!=\!0.75\,\text{m},\,N\!=\!100\times 100$}
\psfrag{data8}[lc][lc][0.7]{$d=\dF,\,D=0.75\,$m}
\psfrag{data1}[lc][lc][0.7]{Threshold}
\centerline{
\includegraphics[width=0.85\linewidth,draft=false]
{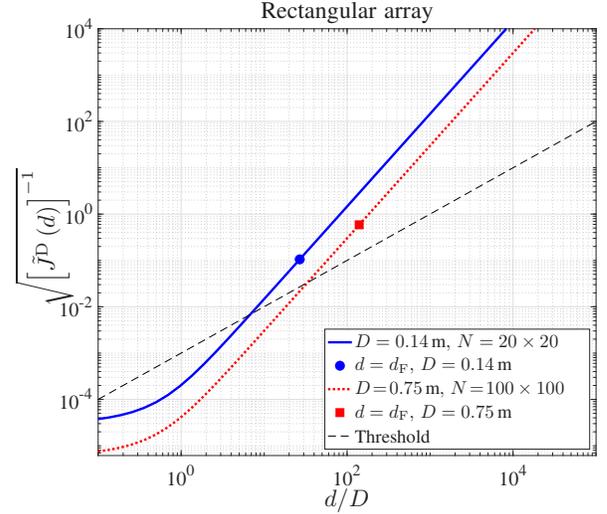}
}
\caption{Ranging error for rectangular arrays and for different array sizes as a function of $d/D$. The central frequency is $28\,$GHz. The measurement noise standard deviation is fixed to $20^\circ$. The number of antennas are $N=20 \times 20$ (blue continuous curve) and $N=100\times 100$ (red dotted curve). The threshold line indicates a ranging error of $0.1\, \%$ of the actual distance.}
\label{fig:crbd_rect}
\end{figure}
Figure~\ref{fig:crbd_rect} reports the square root of the inverse of the ranging \ac{FIM}  
as a function of the normalized distance for rectangular arrays. We notice that
the achieved performance is similar to that obtained for circular arrays because it is driven by the ratio $d/D$, where $D$ is the same in the two settings.  
\end{rem}

\section{Tracking Algorithms}
\label{sec:track}

We now provide an overview of some well-known tracking algorithms to assess their performance and their robustness, using the state-space model in \eqref{eq:motionmodel}-\eqref{eq:observationmodel}, with non--linear Gaussian observation model, and considering the \ac{CoA} for positioning.

\subsection{Extended Kalman Filter}
Among the Bayesian estimators, we start by describing the \ac{EKF}  accounting for the \ac{CoA} information in \eqref{eq:observationmodel}. 
The state is described by a Gaussian distribution, i.e., $\ssk \sim \mathcal{N}\left(\ssk;\bmk, \Pk \right)$, with $\bmk$ and $\Pk$ being the posterior mean vector and covariance matrix of the state. The major steps are reported in Algorithm~\ref{alg:EKF}, and can be described as follows \cite{sarkka2013bayesian}

\noindent \textbf{Initialization}: The \ac{EKF} is initialized by a prior distribution of the state, i.e., $\bsso \sim \po= \mathcal{N}\left(\bsso; \bmo, \Po \right)$.

\begin{algorithm}[t!]
\SetAlgoLined
\textbf{Initialization for $k=1$}:\\
Initialize the state $\bsso \sim  \mathcal{N}\left(\bsso; \bmo, \Po \right)$; \\
  \For{$k=1, \ldots, K$}{
  \textbf{Measurement update};\\
   Calculate the innovation and its covariance; \\
   $\vk=\zk-h\left(\bmkkmo\right)$; \\
   $\Sk=\Hk\Pkkmo\Hk^{\tra} + \Rk$; \\
   Compute the Kalman gain; \\
   $\Kk=\Pkkmo \Hk^{\tra}\,\Sk^{-1}$; \\
   Update the posterior state estimate and covariance; \\
   $\bmkk=\bmkkmo + \Kk\,\vk$; \\
    $\Pkk=\Pkkmo-\Kk\,\Sk\,\Kk^{\tra}$;\\
  \textbf{State Estimation};\\
  Estimate the state; \\ 
  $\hssk=\bmkk$; \\
  \textbf{Time Update}; \\
  Predict the prior state estimate and covariance; \\ 
  $\bmkpok = \A \, \bmkk$; \\
$\Pkpok = \A\, \Pkk\, \A^{\tra} + \Q$;\\
 }
 \caption{Extended Kalman Filter}
 \label{alg:EKF}
\end{algorithm}
\noindent \textbf{Measurement update}: The \ac{EKF} requires the evaluation of the Jacobian matrix associated to the linearization of the observation model $h \left(\psk \right)$ \cite{sarkka2013bayesian}, that can be written as
\begin{align}\label{eq:jacobian}
\Hk 
    &\triangleq \nabla_{\ssk}\, h \left(\psk \right),  
    \end{align}
where  $\nabla_{\ssk}$ is the gradient with respect to the state vector and where the $n$th row of $\Hk$ is given by
\begin{align}
\left\{\Hk\right\}_n
    &=\nabla_{\ssk} \hn\left(\psk \right)=\frac{2 \pi}{\lambdap}\,\nabla_{\ssk}\, \Delta \dnk\left( \psk\right) \,,
\end{align}
with $n=0, \ldots, N-1$ and where $\left\{ \cdot \right\}_n$ picks the $n$th row of $\Hk$ (refer to Appendix A). The Jacobian is evaluated at $\ssk=\bmkkmo$ where $\bmkkmo$ is the predicted state (for $k=1$, it is $\bmkkmo=\mathbf{m}_0$).
Then, following Alg.~\ref{alg:EKF}, the innovation mean and covariance ($\vk$, $\Sk$) and the Kalman gain $\Kk$ are computed and used to update the posterior mean vector $\bmkk$ and covariance matrix $\Pkk$ \cite{sarkka2013bayesian}.

\noindent \textbf{Time update}:
The \ac{EKF} prediction step makes use of the transition model in \eqref{eq:motionmodel}, leading to the estimation of the expected conditional mean $\bmkkmo$ and covariance matrix $\Pkkmo$ for the next time instant.

Unfortunately, the considered differential phases present strong non-linearities {(see Fig.~\ref{fig:phase_ex})} that can hardly be handled by the \ac{EKF} itself, so in the next we will approach the tracking problem by using particle filtering suitable for any transition and observation densities.

\subsection{Particle Filter}
A \ac{PF} exploits the representation of an arbitrary \ac{PDF} by a set of particles and associated weights and where a central role plays a sequential \ac{SIR}  procedure. 
This approach is especially useful for nonlinear non-Gaussian models \cite{DjuEtAl:J03,DjuVemBug:J08}. The goal is the sequential estimation of a filtering distribution, i.e., $p\left(\ssk|\zotok\right)$.
Indeed, this distribution cannot be analytically solved apart from very few cases and, thus, the common procedure  is to exploit discrete random measures composed of particles and weights $\left\{\sskm, \weightkm \right\}_{m=1}^M$, that are possible values of the unknown state $\ssk$, where $M$ is the number of particles \cite{li2015resampling}.
Then, the \ac{PF} can be described by following three major steps reported in Algorithm~\ref{alg:PF}.

\noindent \textbf{Sampling step}: 
The first step is the generation of new particles for the time instant $k$. The particles are drawn from an \ac{IS} density $\pi\left( \cdot \right)$ as $\sskm \sim \pi\left( \ssk \lvert \psskm, \zotok \right)$.
In the sequel, we will review possible choices for the \ac{IS} density.

\noindent \textbf{Importance step}:
Subsequently, the weights $\weightkm$ associated with each particle are computed and normalized. 
The estimate of the state $\hssk$ is inferred as a weighted sum of particles.
\begin{algorithm}[t!]
\SetAlgoLined
\textbf{Initialization for $k=1$}:\\
Initialize the particles $\ssom  \sim p_0$ and their weights $w_{m, 1|0}=1/M, \,\, \forall m$; \\
  \For{$k=1, \ldots, K$}{
  \textbf{Measurement update};\\
  Update the weights according to the likelihood; \\
  $\weightkm\triangleq\weightkkm=\weightkkmom p\left({\zk} \lvert \sskm  \right)$; \\
  Normalize the weights; \\ 
  $\weightkkm= \weightkkm/\sum_m^M \weightkkm $;\\
  \textbf{State Estimation};\\
  Estimate the state;; \\ $\hssk=\sum_m^M \weightkkm \, \sskm$; \\
  \textbf{Resampling};\\
  Resample using multinomial resampling;\\ 
  \textbf{Time Update}; \\
  Predict particles and weights according to the proposal  in \eqref{eq:pproposal}; \\
  $\sskmpo \sim \pi\left(\fssk \lvert \sskm, \fzk \right)$\\
  $\weightkpokm=\weightkkm \frac{p\left(\sskmpo \lvert \sskm \right)}{\pi\left(\sskmpo \lvert \sskm, \fzk \right)}$;
 }
 \caption{Particle Filter}
 \label{alg:PF}
\end{algorithm}
\noindent\textbf{Resampling}:
Finally, to avoid the \textit{degeneracy} problem \cite{li2015resampling} where few particles are dominant, a \textit{resampling} strategy is typically adopted. 
Resampling permits particles with large weights to dominate over particles with small weights, so that at the next time instant, new particles will be generated in the region where large weights are present. After the resampling, weights are also set to be equiprobable, i.e., to $1/M$.

\acp{PF} have become a popular approach because of their ability to operate with models of any nonlinearity and with any noise distributions for as long the likelihoods and the transition pdf that arise from the model \eqref{eq:motionmodel}--\eqref{eq:observationmodel} are computable. 
However, their computational complexity may be high if the number of particles becomes very large
\cite{li2015resampling,doucet2000sequential}.
\subsection{The \acf{IS} density}
The choice of the proposal distribution is one of the most crucial and critical tasks when implementing \ac{PF}s. We now describe some options that perform differently according to the quality of the adopted models.
\paragraph{IS from the prior}
In this case, the \ac{IS} density 
is set equal to the transition distribution function, which is
\begin{align}\label{eq:pproposal}
\sskm& \sim p\left( \sskm \lvert \sskmo \right)=\mathcal{N}\left(\sskm; \A\, \sskmo, \Q\right),
\end{align}
with $\sszm \sim p_0=p\left(\ssz \right)$ being the prior information on the state. As a result, the weights of the particles can be computed using the \ac{LF} of each particle, i.e.,
\begin{align}\label{eq:wup}
&\weightkm = \weightkmo  \,p \left(\zk|\sskm \right). 
\end{align}
A major drawback of this solution is that particles are propagated without taking into consideration the newest measurements $\zk$. Even though the latest measurement is not used for generating new particles, PFs work surprisingly well in most settings. One exception is when the likelihood of the particles is very sharp in comparison of the prior.

\paragraph{IS from the likelihood}

In situations where the likelihood function is much more informative than the prior distribution, a possible alternative to prior \ac{IS} is IS from the likelihood, that is, where particles are generated directly from the likelihood. 
A possibility is to run a \ac{MLE} at each time instant $k$, that, differently from Bayesian approaches, takes only into account the observation model and the latest set of measurements, while neglecting any statistical information regarding the state and its transition model.
More specifically, the state is inferred by solving the following maximization problem:
\begin{align}\label{eq:mle}
&\hbsml = \arg \underset{\ssk}{\max}\, \ln p\left( \bzk \lvert \ssk \right)\,,
\end{align}
where, since the measurements are considered independent at each antenna, we have 
\begin{align}\label{eq:likelihood}
&p\left(\bzk |\ssk \right) = \prod_{n=0}^{N-1} p\left(\znk|\ssk \right)= \prod_{n=0}^{N-1} \mathcal{N}\left(\znk; \hn\left(\psk\right), \sigmaketa^2 \right) \nonumber \\
&= \frac{1}{\sqrt{2\, \pi}\,\sigmaketa}\,  \text{exp}\!\left(- \frac{\sum_{n=0}^{N-1} \left( \znk\! - \!\hn\left(\psk\right)\right)^2 }{2\, \sigmaketa^2 }\right)\!.
\end{align}
Then, the particles are generated from a Gaussian distribution centered at the \ac{ML} estimate by 
\begin{align}\label{eq:lproposal}
\sskm \sim \mathcal{N}\left(\sskm;\hbsml, \BSml \right)  \,,
\end{align}
where $\BSml$ is the covariance matrix that determines how the particles are spread around the \ac{ML} estimate. 
With such choice of IS,  the weights are updated by 
\begin{align}\label{eq:lwupdate}
\weightkm=\weightkmo \frac{p\left({\zk} \lvert \sskm \right)p\left(\sskm\lvert \sskmo \right)}{\mathcal{N}\left(\sskm;\hbsml, \BSml \right)} \,.
\end{align}
Notably, in the considered approach it might happen that $\weightkm\approx0,\,\forall \, m$,
due to the mismatch between the likelihood $p\left({\zk} \lvert \sskm \right)$ and the transition $p\left(\sskm\lvert \sskmo \right)$ in \eqref{eq:lwupdate}. To overcome such an issue, we included a control such that, if all the weights are zeros (or below a certain threshold related to the numerical accuracy), we reset them to $\weightkm=1/M,\,\forall\,m$.

\paragraph{Optimal \ac{IS}} The use of the transition density as an importance function may create ambiguity problems because it does not depend on the new measurements $\zk$. At the same time, the likelihood \ac{IS} does not account for the transition model. Consequently, a possible choice for the optimal \ac{IS} is to directly sample from the posterior \cite{gustafsson2010particle,doucet2000sequential,bunch2013particle} 
\begin{align}
\pi\left( \ssk \lvert \sskmo, \zk\right) 
&= \frac{ p\left(\zk \lvert \ssk\right)  p\left( \ssk \lvert \sskmo\right)}{ \int p\left( \zk \lvert \ssk \right)  p\left( \ssk \lvert \sskmo\right) d \ssk},
\end{align}
where an analytical form can be found  if the observation function is linear and the noises in the state and observation equations are Gaussian and additive.

\paragraph{Local linearisation of the optimal \ac{IS}} In our case, since the observation function is nonlinear, we perform a local linearisation around the predicted state, as done for the \ac{EKF},  with the purpose of deriving a closed-form expression for the proposal density. 
In particular, we have 
\begin{align}\label{eq:liksampTaylor}
&\zk \approx  h\left( f\left( \pssk \right)\right) +  \Hkkmo \left(\ssk -f\left( \pssk \right) \right) + \betak,
\end{align}
where $\Hkkmo \triangleq \Hk \Big\lvert_{\ssk = f\left( \pssk \right)}$ is the $N \times \Ns$ Jacobian matrix in \eqref{eq:jacobian} evaluated at the predicated state, i.e., at $\ssk=f\left( \pssk \right)=\A\, \pssk$. 
Then we can express \eqref{eq:liksampTaylor} as a  function of the state, i.e.,
\begin{align}\label{eq:liksampTaylor2}
 \ssk &=\Hkkmo^{\dagger}\, \left(\zk - h\left( f\left( \pssk \right)\right) - \betak \right) + f\left( \pssk \right) \nonumber \\
&=\Hkkmo^{\dagger}\, \left(\zk - h\left( f\left( \pssk \right)\right)\right) + f\left( \pssk \right) + \tbetak,
\end{align}
where $\dagger$ is the pseudo-inverse operator, $\Hkkmo^{\dagger}$ is the $\Ns \times N$ Moore-Penrose inverse of the predicted Jacobian matrix, $\tbetak \sim \mathcal{N}\left( 0, \tRk\right)$ and $\tRk=\Hkkmo^{\dagger} \,\Rk\, \left(\Hkkmo^{\dagger}\right)^{\tra}$.  Consequently, by considering the product of Gaussian densities, 
%
%
it is possible to derive a density for the state that is
\begin{align}\label{eq:poptimal}
&\sskm \sim p\left( \sskm \lvert \sskmo, \zk\right) 
\approx \mathcal{N}\left( \sskm; \bmumk, \Smk \right),
\end{align}
where the mean and covariance matrix are derived  as \cite{bunch2013particle} 
\begin{align}
\label{eq:covopt}
&\Smk^{-1}=\Q^{-1} + \tRmk^{-1}, \\ 
&\bmumk=\Smk\,\Bigg[ \Q^{-1}\, f\left( \psskm \right) +  \Hmkkmo^{\tra} \, \R^{-1}  \nonumber \\
&\times \bigg( \zk - h\left( f\left(\psskm \right)\right)  + \Hmkkmo \,  f\left( \psskm\right)   \bigg) \Bigg], \label{eq:mopt}
\end{align}
where $\Hmkkmo$ and  $\tRmk$ are computed at the particle states.
In this case, the weights associated with each particle are obtained by 
\begin{equation}\label{eq:w11}
    \weightkm= \weightkmo\, \frac{p\left(\zk \lvert \sskm \right)p\left(\sskm\lvert \sskmo \right)}{\mathcal{N}\left( \sskm; \bmumk, \Smk \right)}.
\end{equation}

The optimal \ac{IS} represents a trade-off between the prior and the likelihood \ac{IS} and it provides good performance when both models are  accurate. 
In the following, we evaluate and compare their performances.

\section{Case Study}\label{sec:results}

\subsection{Simulation Parameters}
\label{sec:casestudy_sp}
We now evaluate the tracking performance and the theoretical bound by varying the array size and the model parameters.
To this purpose, we set $\lambda=0.01\,$m, and the number of particles to $M=1000$, if not otherwise indicated.

A large array was placed
in the origin, i.e., the reference location was $\pok \triangleq \left[\xo,\, \yo,\, \zo \right]^{\tra}\, \text{(m)}=\left(0,0,1 \right)$, and we alternatively considered a planar rectangular array  lying on the $YZ$-plane with $N=20 \times 20$ or $N=30\times 30$ antennas. 

The initial state of the target at time instant was $\ssz \triangleq \left[x_{0},\, y_{0},\, z_{0}\, v_{\mathsf{x},0}\, v_{\mathsf{y}, 0}\, v_{\mathsf{z}, 0} \right]^{\tra} = \left(2.5,\, -9.1,\, 1.5,\, 0.01,\, 0.97,\, 0 \right)$, where the simulation step was fixed to $\tau=1\,$second, the position and velocity coordinates were in (m) and (m/step), respectively. The total number of time instants was $K=20\,$.

The actual transition of the source followed the linear model in \eqref{eq:motionmodel} with the transition function and covariance matrix set to have a nearly constant velocity movement according to
\begin{align}\label{eq:tmodel}
&\A=\left[\begin{array}{cc} \mathbf{I}_3 & \tau\, \mathbf{I}_3 \\
\mathbf{0}_3 & \mathbf{I}_3
\end{array} \right], &&\Q=\left[\begin{array}{cc}
\frac{\tau^3}{3} \, \Q_{\mathsf{a}}& \frac{\tau^2}{2} \, \Q_{\mathsf{a}}\\
\frac{\tau^2}{2} \, \Q_{\mathsf{a}} & \tau \, \Q_{\mathsf{a}}
\end{array} \right],
\end{align}
where   $\Q_{\mathsf{a}}$ is a diagonal matrix containing the variances of the change in accelerations, i.e., $\Q_{\mathsf{a}} = \operatorname{diag}\left(\sigma^2_{\mathsf{a},\mathsf{x}}, \sigma^2_{\mathsf{a},\mathsf{y}} ,\sigma^2_{\mathsf{a},\mathsf{z}}  \right)$, where $\sigma^2_{\mathsf{a},\mathsf{x}}=\sigma^2_{\mathsf{a},\mathsf{y}}=\gammat\, 0.03^2\, \left(\text{m}^2/\text{step}^6\right)$, {with $\gammat=1$}, and $\sigma^2_{\mathsf{a},\mathsf{z}}=0$. {Instead, for the tracking estimator}, we considered alternatively $\gammat=1$ and $\gammat=10$, that represented the possibility to work with a transition model that was the same as the one used for the actual target trajectory (transition parameter match - $\textsf{TM}_0$) or not (transition parameter mismatch - $\textsf{TM}_1$), respectively. 
The measurements were
generated using the model described by
\eqref{eq:observationmodel}-\eqref{eq:observationmodel_v2}, where the noise standard deviation was set to
$\sigmaketa=\sigma \cdot (1+ \gammam)$ with $\sigma=20^\circ$ (if not otherwise indicated) and where $\gammam=0$ (i.e., $\sigmaeta=20^\circ$) and $\gammam=1$ (i.e., $\sigmaeta=40^\circ$) denote a model parameter match (measurement parameter match - $\textsf{MM}_0$) or mismatch (measurement parameter match - $\textsf{MM}_1$), respectively.
\begin{figure}[t!]
\psfrag{x}[c][c][0.7]{$x$ [m]}
\psfrag{y}[c][c][0.7]{$y$ [m]}
\psfrag{0}[c][c][0.7]{$0$}
\psfrag{0.1}[c][c][0.7]{}
\psfrag{0.2}[c][c][0.7]{$0.2$}
\psfrag{0.3}[c][c][0.5]{}
\psfrag{0.4}[c][c][0.7]{$0.4$}
\psfrag{0.5}[c][c][0.5]{}
\psfrag{0.6}[c][c][0.7]{$0.6$}
\psfrag{0.7}[c][c][0.5]{}
\psfrag{0.8}[c][c][0.7]{$0.8$}
\psfrag{0.9}[c][c][0.7]{}
\psfrag{1}[c][c][0.7]{$1$}
\psfrag{1.5}[c][c][0.5]{}
\psfrag{2}[c][c][0.7]{$2$}
\psfrag{2.5}[c][c][0.5]{}
\psfrag{3}[c][c][0.7]{$3$}
\psfrag{3.5}[c][c][0.5]{}
\psfrag{4}[c][c][0.7]{$4$}
\psfrag{4.5}[c][c][0.5]{}
\psfrag{5}[c][c][0.7]{$5$}
\centerline{
\psfrag{t}[c][c][0.7]{$4\times 4$, $d=2.15$ m (12.5 $\dF$)}
\includegraphics[width=0.85 \linewidth,draft=false]
{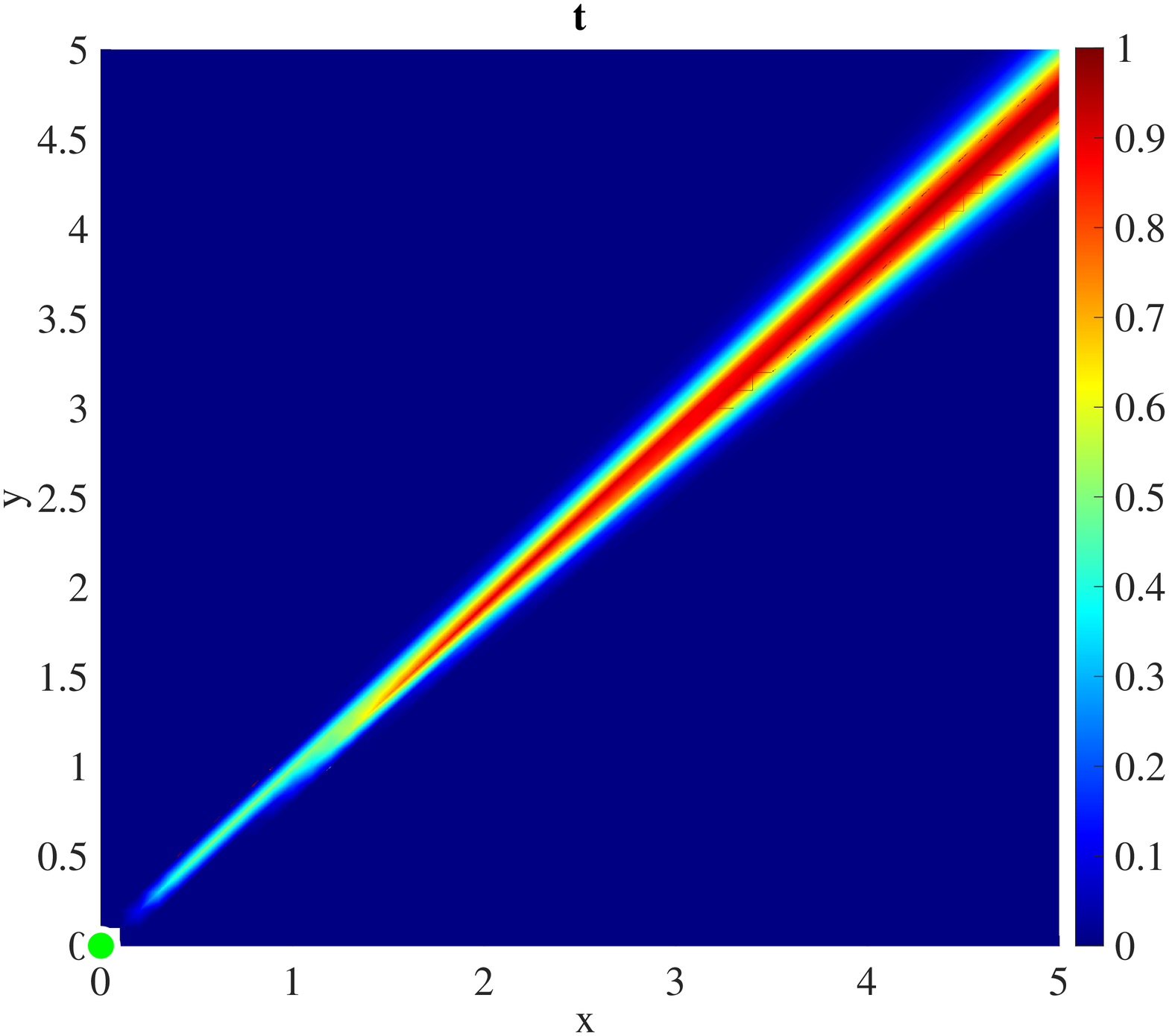}} \vspace{0.2cm}
\centerline{
\psfrag{t}[c][c][0.7]{$20\times 20$, $d=2.15$ m (0.5 $\dF$)}
\includegraphics[width=0.85 \linewidth,draft=false]
{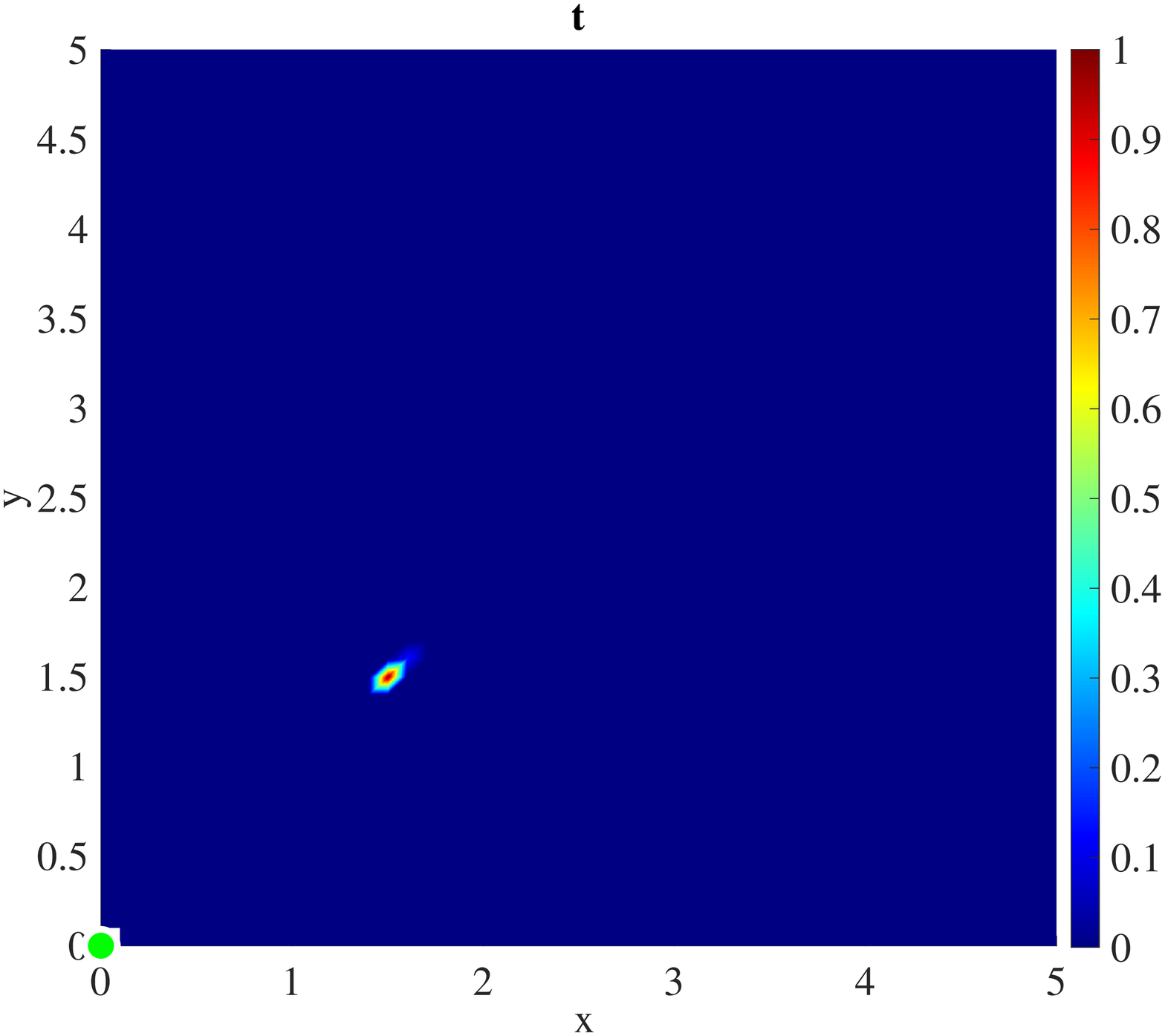}
}
\caption{Normalized \ac{LF} for planar arrays with $4 \times 4$ (top) and $20 \times 20$ (bottom) antennas on the $YZ-$ plane, with $\sigmaeta=20^\circ$. 
The receiver and target locations were in $
\left[0,0,1\right]$ and in $\left[1.51,\,1.51,\, 1\right]$, respectively.
}
\label{fig:ML_1}
\end{figure}

The \ac{EKF} and the particles were initialized according to
\begin{align}
&\mathbf{m}_0=\mathbf{s}_{m, 0}=\mathcal{N}\left(\mathbf{s}_{0}, \bm{\Sigma}_0 \right), \\
&\mathbf{P}_0 \!=\! \bm{\Sigma}_0 
\!=\! \operatorname{diag}\left(0.5^2, 0.5^2 ,0.01^2, \frac{v_{\mathsf{x},0}^2}{100} , \frac{v_{\mathsf{y},0}^2}{100} , \frac{v_{\mathsf{z},0}^2}{100}  \right)\,,
\end{align}
if not otherwise indicated. In  the \ac{PF} method, we exploited the multinomial resampling  strategy \cite{li2015resampling}. For the likelihood \ac{IS}, we set $\bm{\Sigma}_{\mathbf{s}}=\bm{\Sigma}_0, \forall k$.

For the \ac{MLE}, we used a scatter search algorithm implemented in
MATLAB~R2020a software ($\operatorname{GlobalSearch}$ command)
to find a global minimum \cite{ugray2007scatter,glover1998template} 

\subsection{Numerical results}

\subsubsection{Likelihood function}
We first investigated the \ac{LF} shape considering the observation model in \eqref{eq:observationmodel_v2}   and a target located inside and outside the Fresnel region, delimited by $\dF$.
To that end, we considered a $5\times 5\,$m$^2$ grid of points equally spaced with a step of $0.1\,$m corresponding to the state $\ssk^i$, with $i$ being the index of the $i$th grid point and $k$ the time instant. For each test position, we computed the \ac{LF} related to the actual state. 

Fig.~\ref{fig:ML_1} shows the normalized \ac{LF}  for  $4\times 4$ and  $20\times 20$ arrays. The target was located at a distance of $d=2.15\,$m from the array that corresponds to $0.5\,\dF$ for the $20\times 20$ array and $12.5\,\dF$ for the $4\times 4$ array. We notice that the \ac{LF} is peaky and focused on the target's position when a large array is used as the target falls in its near--field region, whereas it becomes less and less sharp and with ambiguities when exiting the Fresnel region because the effect of the \ac{CoA} tends to vanish. Nevertheless, as demonstrated in Sec~\ref{sec:tracklim}, there is no variation in the performance of the angle estimation when moving from the near--field to the far--field region, as it is also evident from the sector shape in Fig.~\ref{fig:ML_1} (top).

\begin{figure}[t!]
\psfrag{x}[c][c][0.7]{$x$ [m]}
\psfrag{y}[c][c][0.7]{$y$ [m]}
\psfrag{10}[c][c][0.5]{$10$}
\psfrag{8}[c][c][0.5]{$8$}
\psfrag{6}[c][c][0.5]{$6$}
\psfrag{4}[c][c][0.5]{$4$}
\psfrag{2}[c][c][0.5]{$2$}
\psfrag{0}[c][c][0.5]{$0$}
\psfrag{-2}[c][c][0.5]{\!\!\!$-2$}
\psfrag{-4}[c][c][0.5]{\!\!\!$-4$}
\psfrag{-6}[c][c][0.5]{\!\!\!$-6$}
\psfrag{-8}[c][c][0.5]{\!\!\!$-8$}
\psfrag{-10}[c][c][0.5]{\!\!\!$-10$}
\psfrag{data1111111111111}[lc][lc][0.5]{Array}
\psfrag{data2}[lc][lc][0.5]{True trajectory}
\psfrag{data3}[lc][lc][0.5]{EKF}
\psfrag{data4}[lc][lc][0.5]{MLE}
\psfrag{data5}[lc][lc][0.5]{PF - P-IS}
\psfrag{data6}[lc][lc][0.5]{PF - L-IS}
\psfrag{data7}[lc][lc][0.5]{PF - LO-IS}
\centerline{
\psfrag{T}[c][c][0.7]{$20\times 20, \sigmaeta=10^\circ$}
\includegraphics[width=0.5 \linewidth,draft=false]
{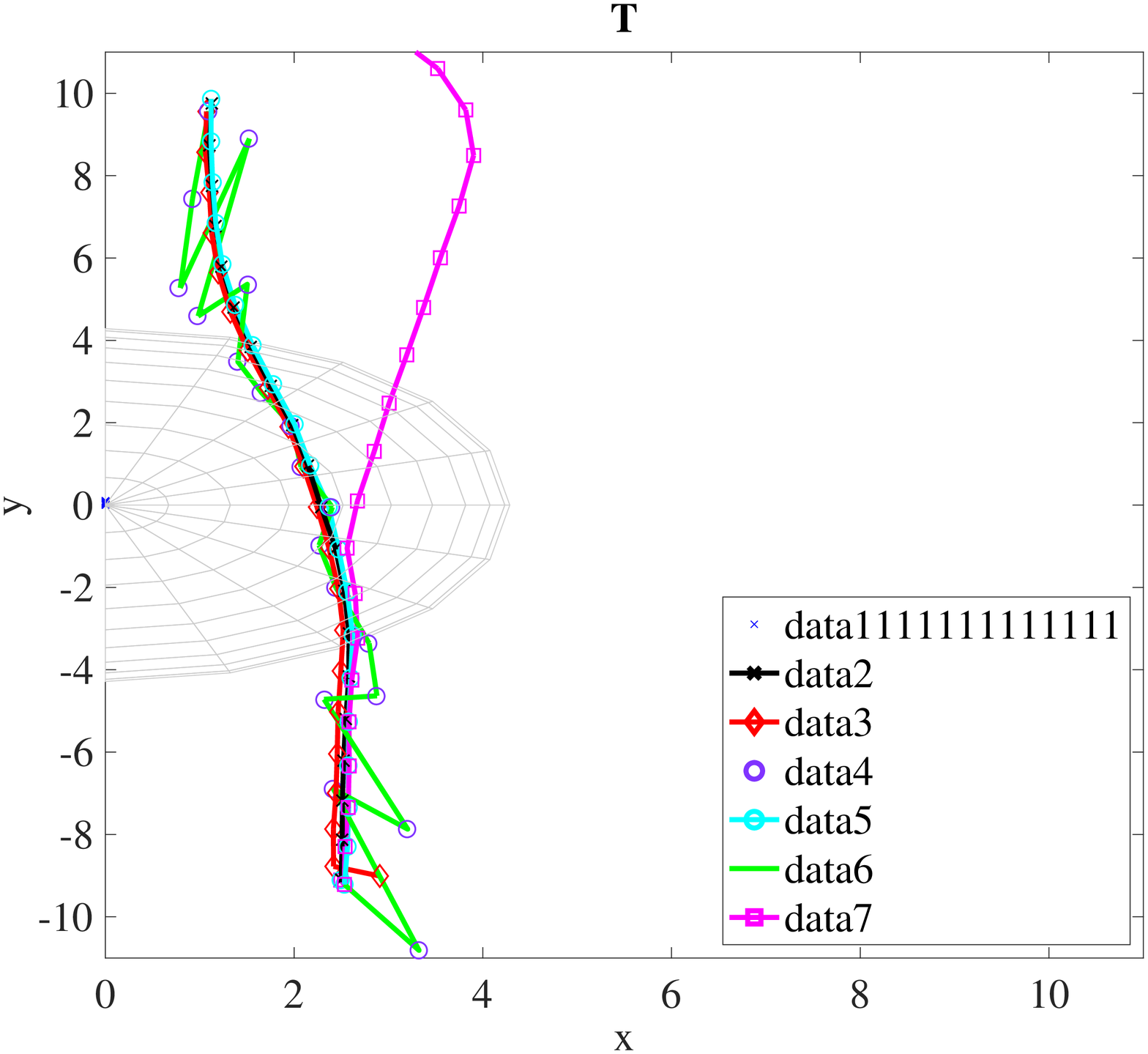}
\psfrag{T}[c][c][0.7]{$20\times 20, \sigmaeta=20^\circ$}
\includegraphics[width=0.5 \linewidth,draft=false]
{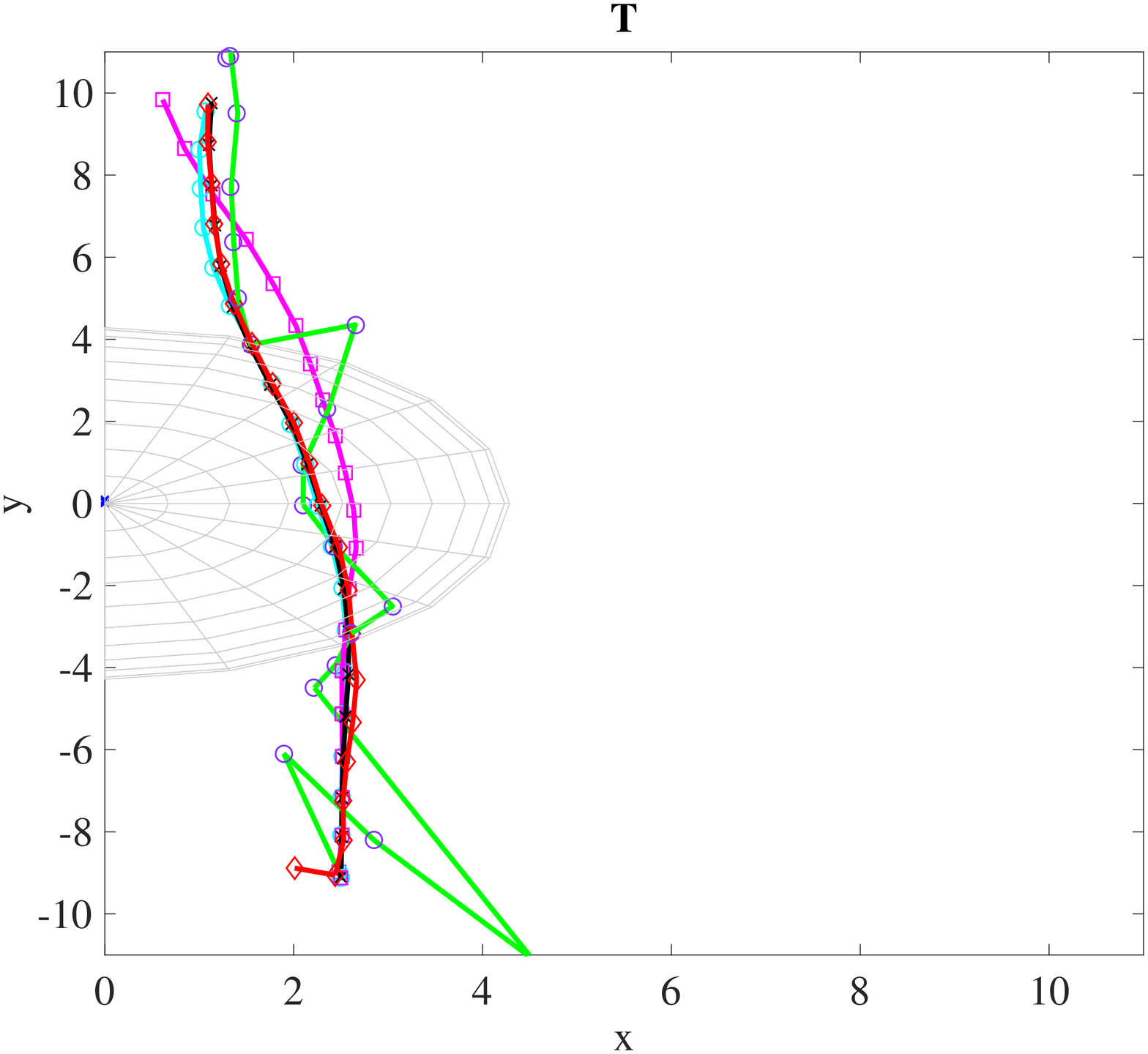}
} 
\vspace{0.2cm}
\centerline{
\psfrag{T}[c][c][0.7]{$30\times 30, \sigmaeta=10^\circ$}
\includegraphics[width=0.5 \linewidth,draft=false]
{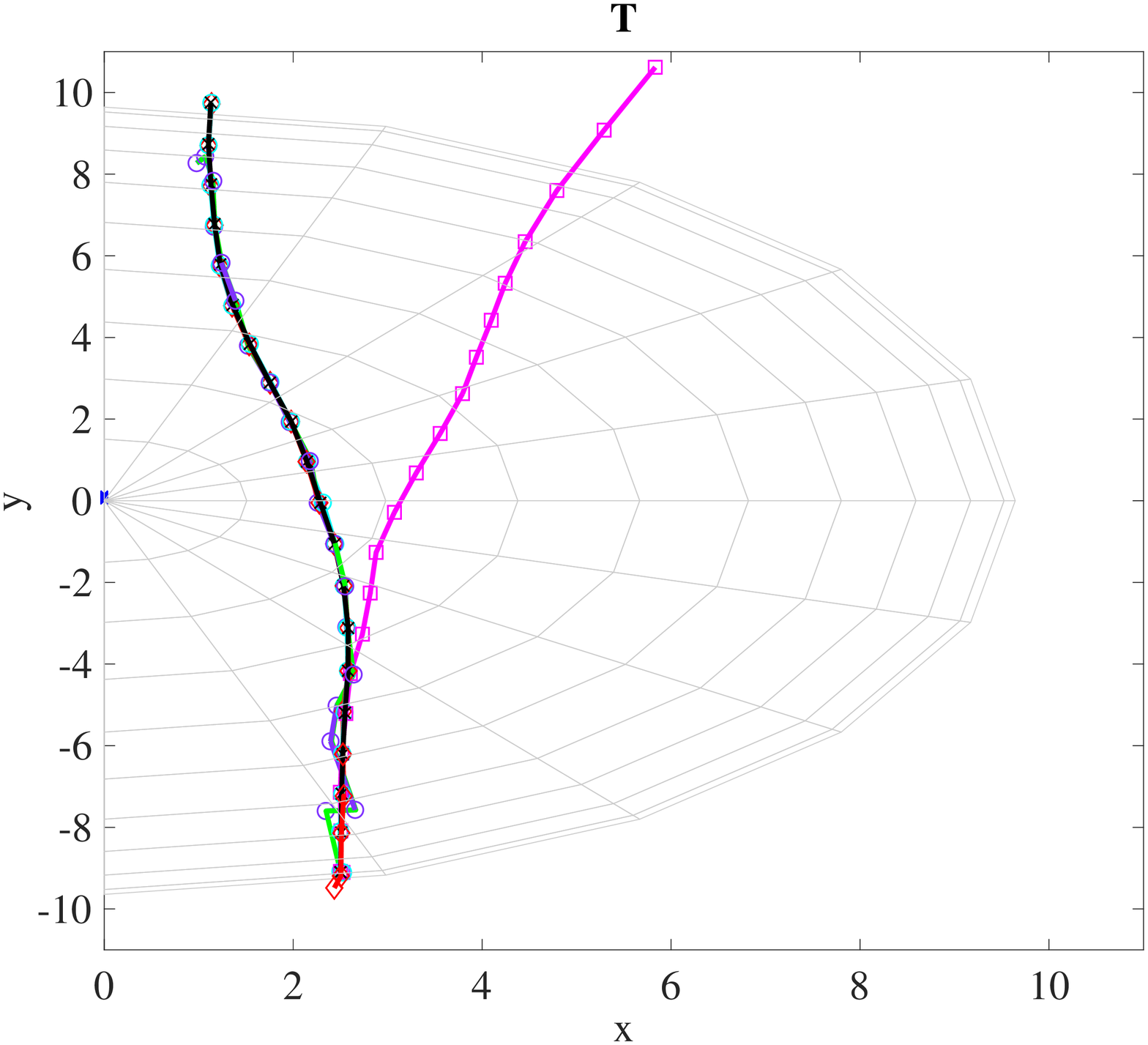}
\psfrag{T}[c][c][0.7]{$30\times 30, \sigmaeta=20^\circ$}
\includegraphics[width=0.5 \linewidth,draft=false]
{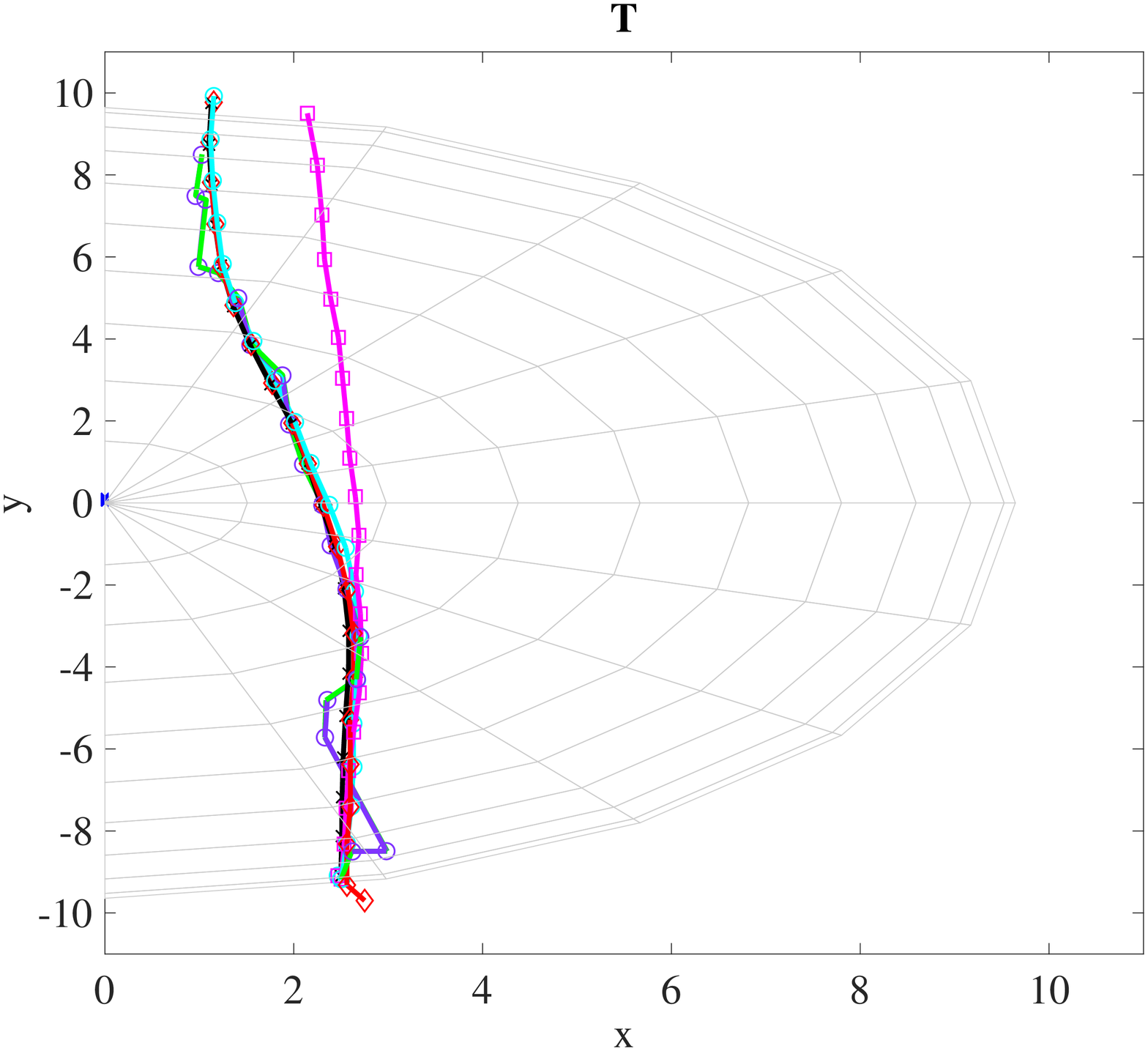}
} 
\psfrag{data1111111111111111}[lc][lc][0.5]{Array}
\psfrag{data2}[lc][lc][0.5]{Actual trajectory}
\psfrag{k}[lc][lc][0.7]{$k=9$}
\psfrag{data3}[lc][lc][0.5]{LO-IS, $\sigmaeta=10^\circ$}
\psfrag{data4}[lc][lc][0.5]{LO-IS, $\sigmaeta=100^\circ$}
\psfrag{x}[c][c][0.7]{$x$ [m]}
\psfrag{y}[c][c][0.7]{$y$ [m]}
\centerline{
\includegraphics[width=0.5 \linewidth,draft=false]
{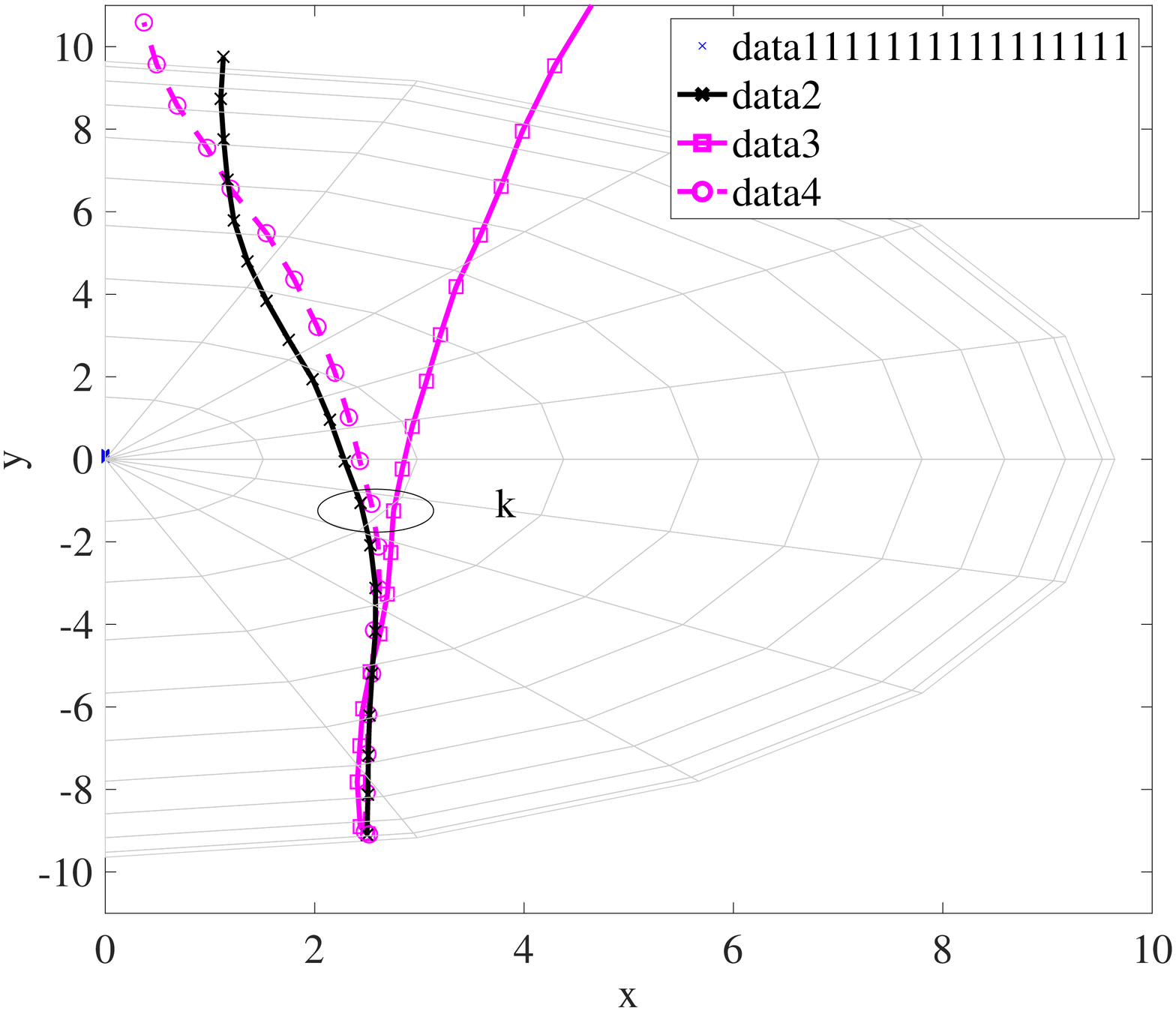}
}
\caption{Example of estimated trajectories for different approaches and array sizes. Top: $N=20\times 20$, $\textsf{TM}_0$, $\textsf{MM}_0$; Middle: $N=30\times 30$, $\textsf{TM}_0$, $\textsf{MM}_0$; Bottom:   $N=30\times 30$, $\textsf{TM}_0$, with and without measurement parameter mismatch ($\textsf{MM}_0$ vs. $\textsf{MM}_1$).  The array reference location is in $\left[0,\,0,\, 1 \right]$ and is lying in the $YZ$ plane. (P-IS) indicates the \ac{PF} with prior \ac{IS}, (LO-IS) is the \ac{PF} with linearised optimal \ac{IS} and (L-IS) is the \ac{PF} with likelihood \ac{IS}. }
\label{fig:trajectories_forCDF}
\end{figure}
\begin{figure}[t!]
\psfrag{x}[c][c][0.7]{$x$ [m]}
\psfrag{y}[c][c][0.7]{\!\!\!\!\!\!\!$y$ [m]}
\psfrag{-1.245}[c][c][0.5]{}
\psfrag{-1.25}[c][c][0.5]{\!\!$-1.25$}
\psfrag{-1.255}[c][c][0.5]{}
\psfrag{-1.26}[c][c][0.5]{\!\!$-1.26$}
\psfrag{-1.265}[c][c][0.5]{}
\psfrag{-1.27}[c][c][0.5]{\!\!$-1.27$}
\psfrag{-1.275}[c][c][0.5]{}
\psfrag{-1.28}[c][c][0.5]{\!\!$-1.28$}
\psfrag{-30}[c][c][0.5]{$-30$}
\psfrag{-40}[c][c][0.5]{$-40$}
\psfrag{-50}[c][c][0.5]{$-50$}
\psfrag{-60}[c][c][0.5]{$-60$}
\psfrag{-70}[c][c][0.5]{$-70$}
\psfrag{-80}[c][c][0.5]{$-80$}
\psfrag{-90}[c][c][0.5]{$-90$}
\psfrag{-100}[c][c][0.5]{$-100$}
\psfrag{t}[c][c][0.6]{Measurement Update, PF - LO-IS}
\centerline{
\includegraphics[width=0.5 \linewidth,draft=false]
{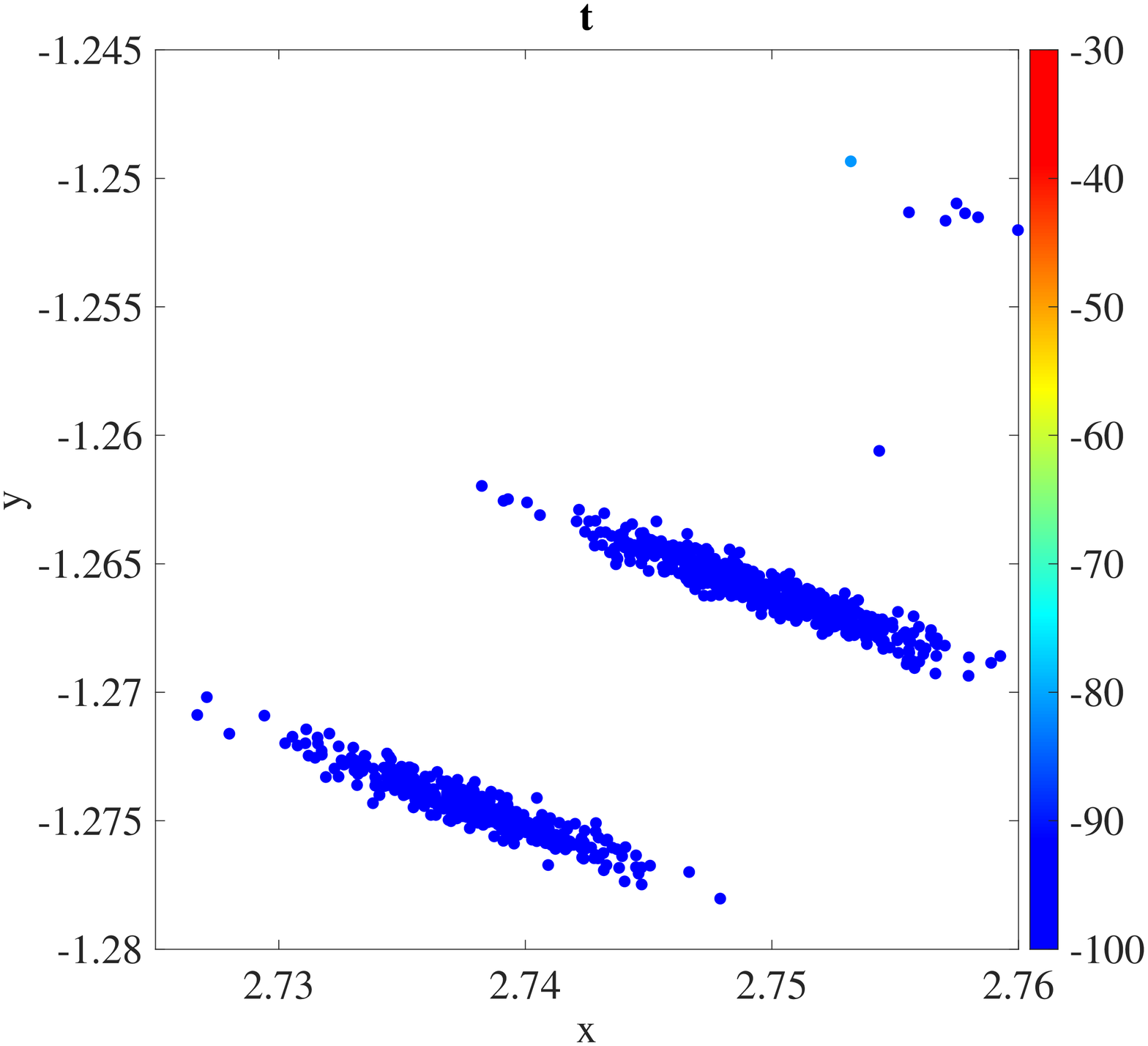}
\psfrag{t}[c][c][0.6]{Time Update, PF - LO-IS}
\includegraphics[width=0.5 \linewidth,draft=false]
{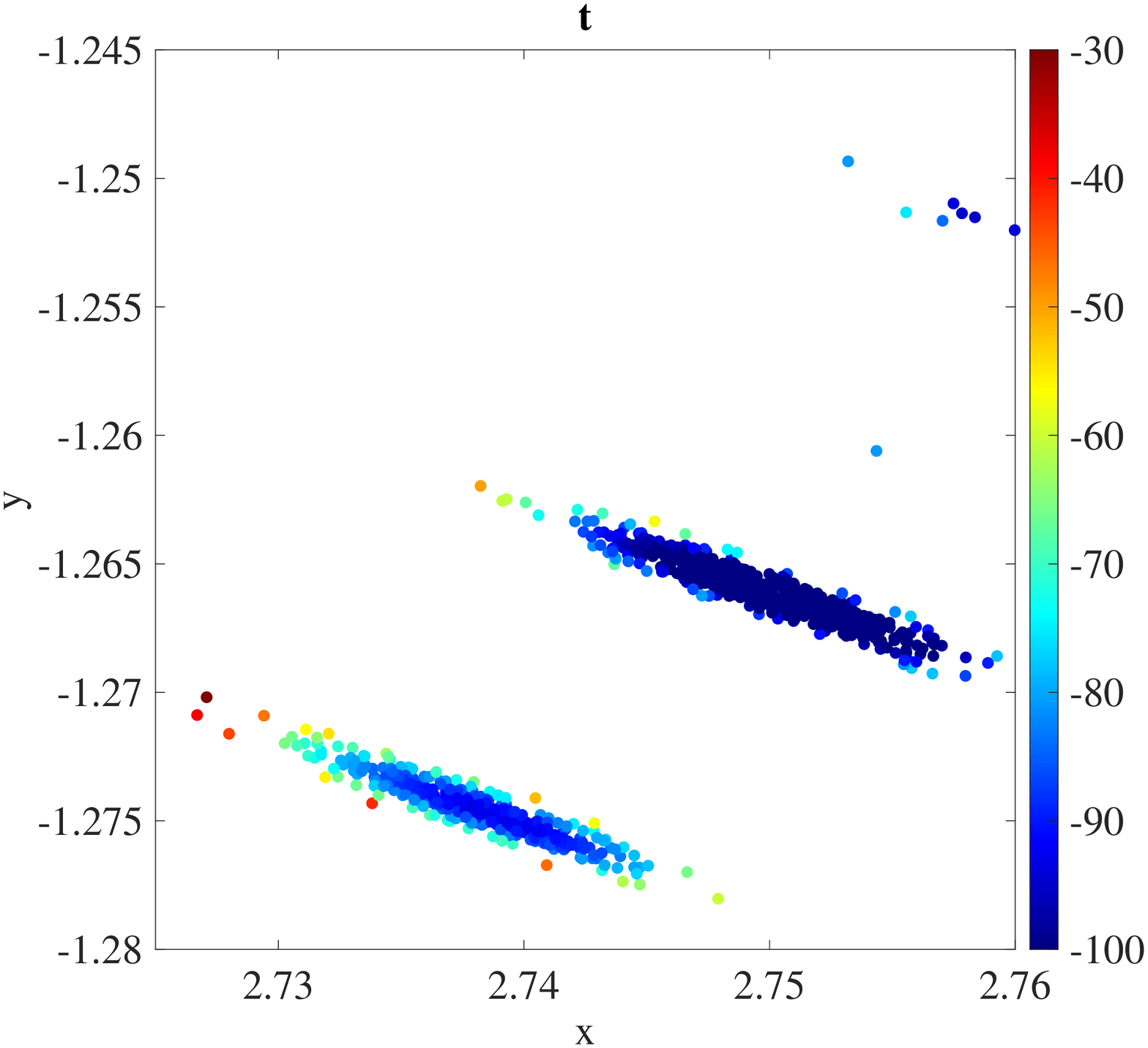}
}  
\vspace{0.2cm}
\psfrag{x}[c][c][0.7]{$x$ [m]}
\psfrag{y}[c][c][0.7]{\!\!\!\!\!\!\!\!$y$ [m]}
\psfrag{-30}[c][c][0.5]{$-30$}
\psfrag{-40}[c][c][0.5]{$-40$}
\psfrag{-50}[c][c][0.5]{$-50$}
\psfrag{-60}[c][c][0.5]{$-60$}
\psfrag{-70}[c][c][0.5]{$-70$}
\psfrag{-80}[c][c][0.5]{$-80$}
\psfrag{-90}[c][c][0.5]{$-90$}
\psfrag{-100}[c][c][0.5]{$-100$}
\psfrag{-1}[c][c][0.5]{\!\!$-1$}
\psfrag{-1.02}[c][c][0.5]{\!\!$-1.02$}
\psfrag{-1.04}[c][c][0.5]{}
\psfrag{-1.06}[c][c][0.5]{\!\!$-1.06$}
\psfrag{-1.08}[c][c][0.5]{}
\psfrag{-1.1}[c][c][0.5]{\!\!$-1.1$}
\psfrag{-1.12}[c][c][0.5]{}
\psfrag{-1.14}[c][c][0.5]{\!\!$-1.14$}
\psfrag{2.4}[c][c][0.5]{$2.4$}
\psfrag{2.45}[c][c][0.5]{$2.45$}
\psfrag{2.5}[c][c][0.5]{$2.5$}
\centerline{
\psfrag{t}[c][c][0.6]{Measurement Update, PF - P-IS}
\includegraphics[width=0.5 \linewidth,draft=false]
{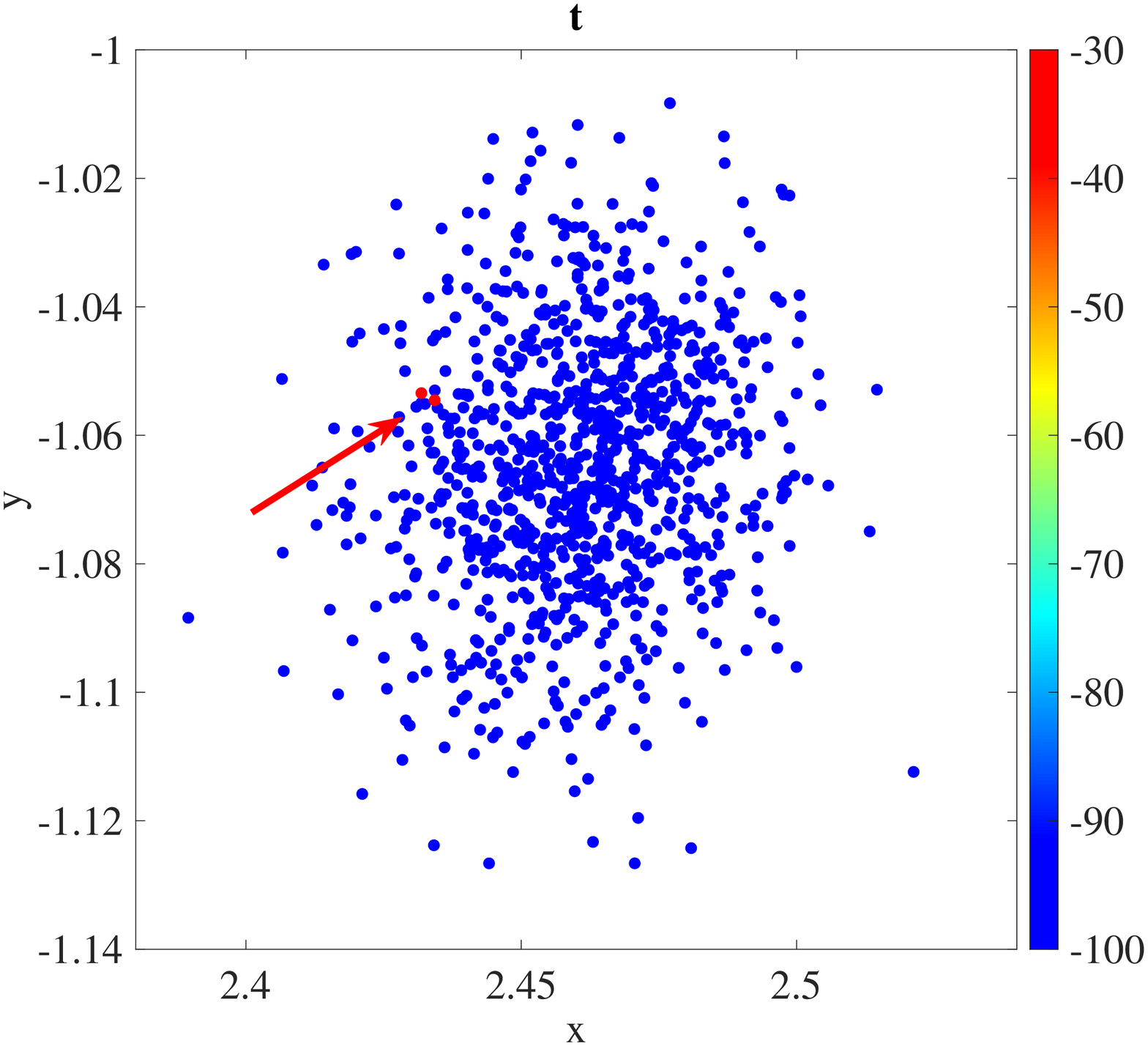}
\psfrag{t}[c][c][0.6]{Time Update, PF - P-IS}
\includegraphics[width=0.5 \linewidth,draft=false]
{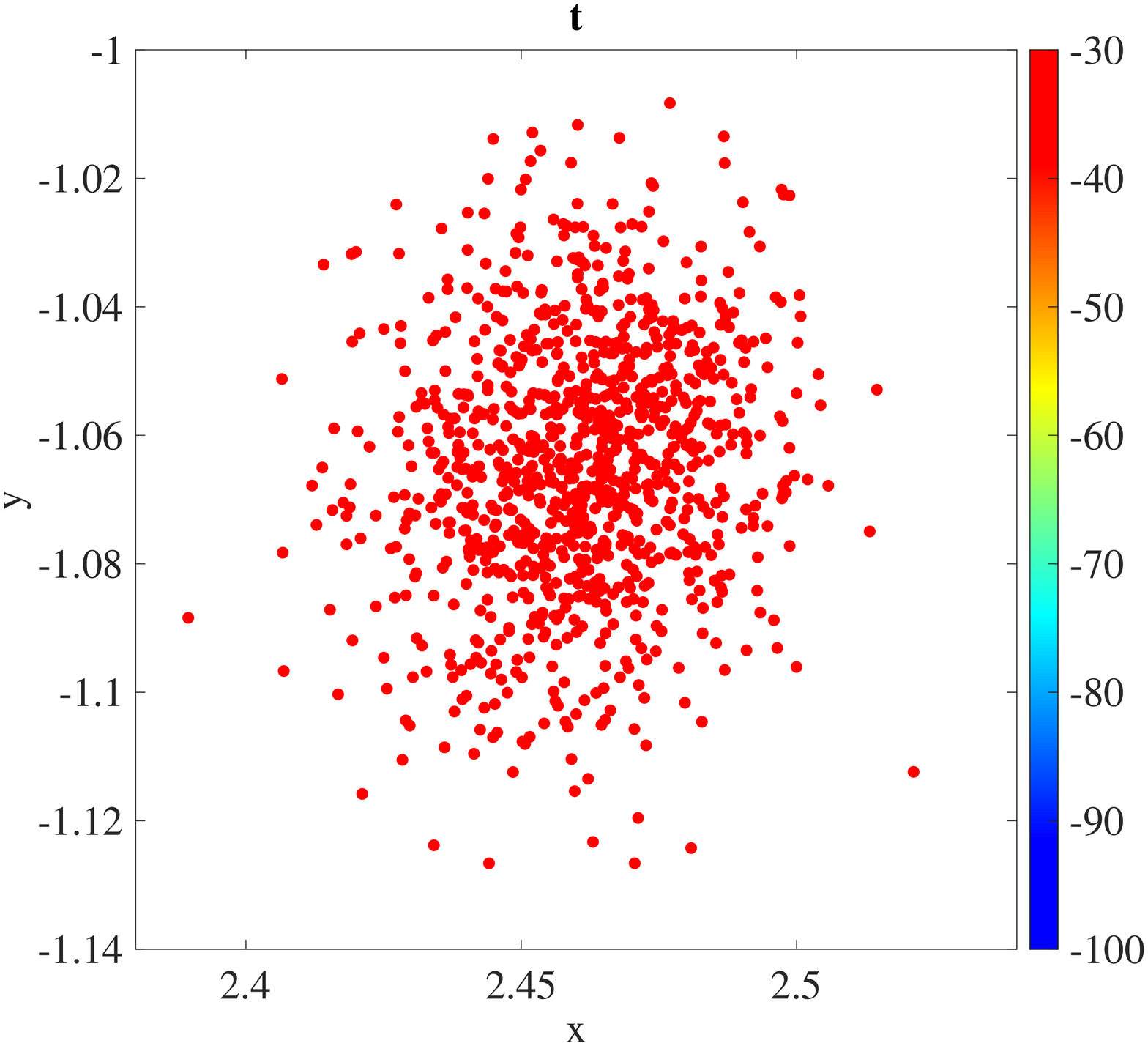}
} \vspace{0.2cm}
\psfrag{x}[c][c][0.7]{$x$ [m]}
\psfrag{y}[c][c][0.7]{\!\!\!\!\!\!\!\!$y$ [m]}
\psfrag{-30}[c][c][0.5]{$-30$}
\psfrag{-40}[c][c][0.5]{$-40$}
\psfrag{-50}[c][c][0.5]{$-50$}
\psfrag{-60}[c][c][0.5]{$-60$}
\psfrag{-70}[c][c][0.5]{$-70$}
\psfrag{-80}[c][c][0.5]{$-80$}
\psfrag{-90}[c][c][0.5]{$-90$}
\psfrag{-100}[c][c][0.5]{$-100$}
\psfrag{-1.245}[c][c][0.5]{}
\psfrag{-1.25}[c][c][0.5]{\!\!$-1.25$}
\psfrag{-1.255}[c][c][0.5]{}
\psfrag{-1.26}[c][c][0.5]{\!\!$-1.26$}
\psfrag{-1.265}[c][c][0.5]{}
\psfrag{-1.27}[c][c][0.5]{\!\!$-1.27$}
\psfrag{-1.275}[c][c][0.5]{}
\psfrag{-1.28}[c][c][0.5]{\!\!$-1.28$}
\centerline{
\psfrag{t}[c][c][0.6]{PF - LO-IS, $\sigmaeta=10^\circ$, $k=9$}
\includegraphics[width=0.5 \linewidth,draft=false]
{./fig11.eps}
\psfrag{t}[c][c][0.6]{PF - LO-IS, $\sigmaeta=100^\circ$, $k=9$}
\includegraphics[width=0.5 \linewidth,draft=false]
{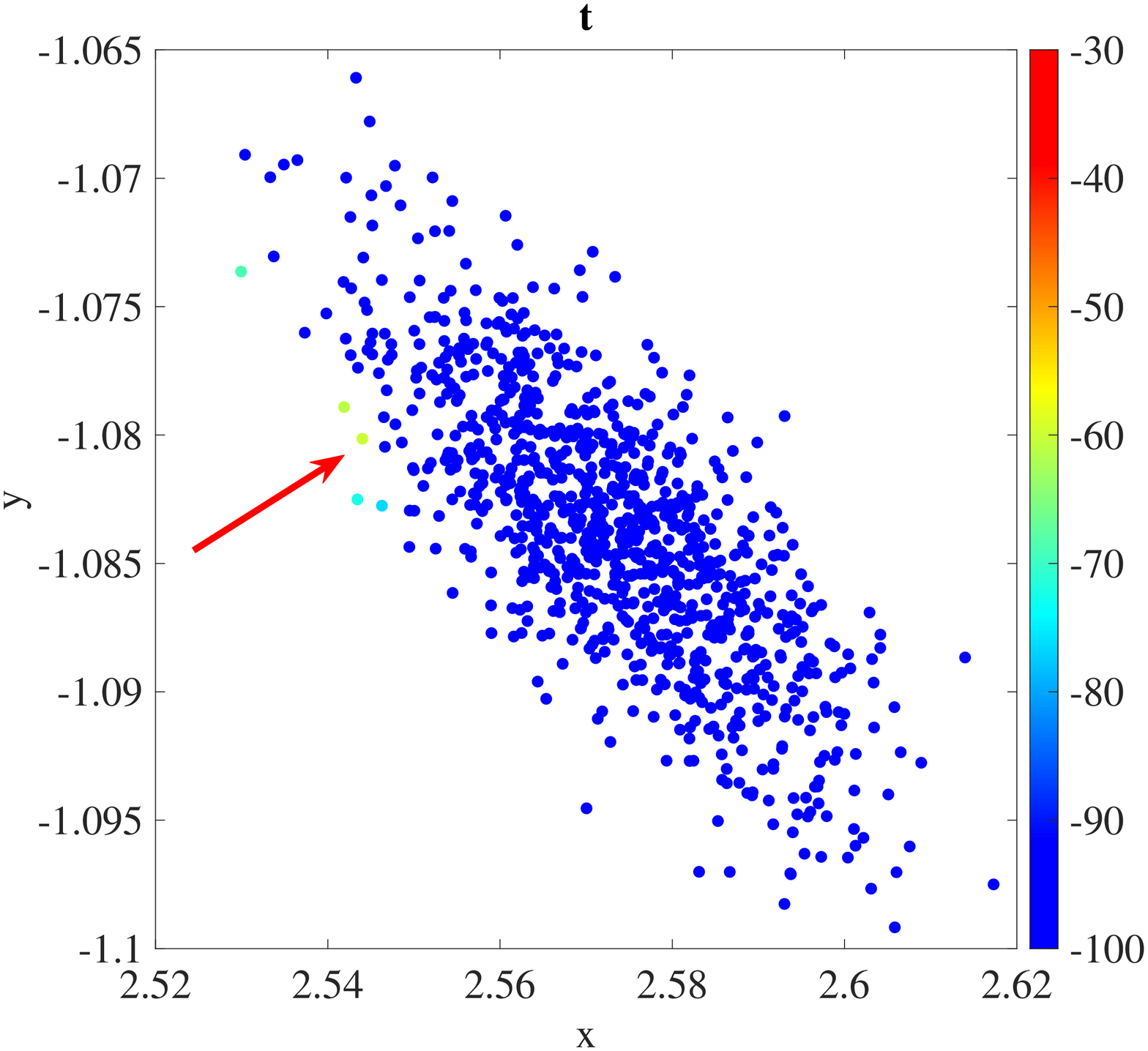}
}
\caption{Particles and weights for a \ac{PF} with linearised optimal \ac{IS} (top, bottom) and for prior \ac{IS} (middle). (Top, Middle)-Left and Bottom: Measurement update at $k=9$. (Top, Middle)-Right: Time update step at $k=8$. (Bottom)-Left: $\textsf{MM}_0$ with $\sigmaeta=10^\circ$. (Bottom)-Right: $\textsf{MM}_1$ with $\sigmaeta=100^\circ$.The red arrow indicates a group of/a single particle/s experiencing larger weights.  }
\label{fig:weightsMM1}
\end{figure}
\begin{figure}[t!]
\psfrag{x}[c][c][0.7]{$x$ [m]}
\psfrag{y}[c][c][0.7]{$y$ [m]}
\psfrag{data11111}[lc][lc][0.5]{Array}
\psfrag{data2}[lc][lc][0.5]{True Traj.}
\psfrag{data3}[lc][lc][0.5]{MLE}
\psfrag{data4}[lc][lc][0.5]{PF - P-IS}
\psfrag{data5}[lc][lc][0.5]{PF - L-IS}
\psfrag{10}[c][c][0.5]{$10$}
\psfrag{8}[c][c][0.5]{$8$}
\psfrag{6}[c][c][0.5]{$6$}
\psfrag{4}[c][c][0.5]{$4$}
\psfrag{2}[c][c][0.5]{$2$}
\psfrag{0}[c][c][0.5]{$0$}
\psfrag{-2}[c][c][0.5]{\!\!\!$-2$}
\psfrag{-4}[c][c][0.5]{\!\!\!$-4$}
\psfrag{-6}[c][c][0.5]{\!\!\!$-6$}
\psfrag{-8}[c][c][0.5]{\!\!\!$-8$}
\psfrag{-10}[c][c][0.5]{\!\!\!$-10$}
\centerline{
\psfrag{t}[c][c][0.7]{$30\times 30$ $\sigmaeta=20^\circ$, $\textsf{TM}_0$}
\includegraphics[width=0.5 \linewidth,draft=false]
{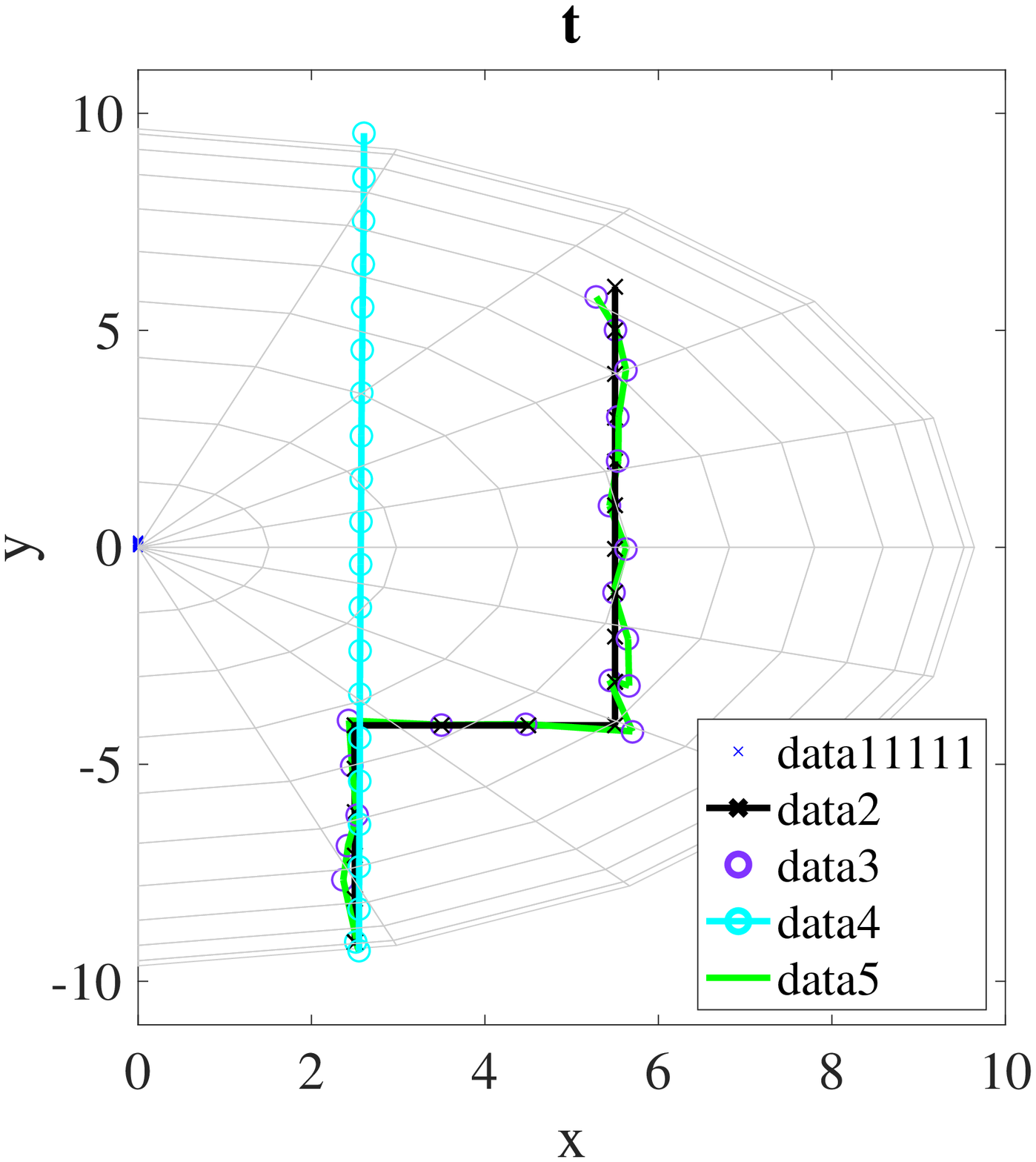}
\psfrag{t}[c][c][0.7]{$30\times 30$ $\sigmaeta=20^\circ$, $\textsf{TM}_1$}
\includegraphics[width=0.5 \linewidth,draft=false]
{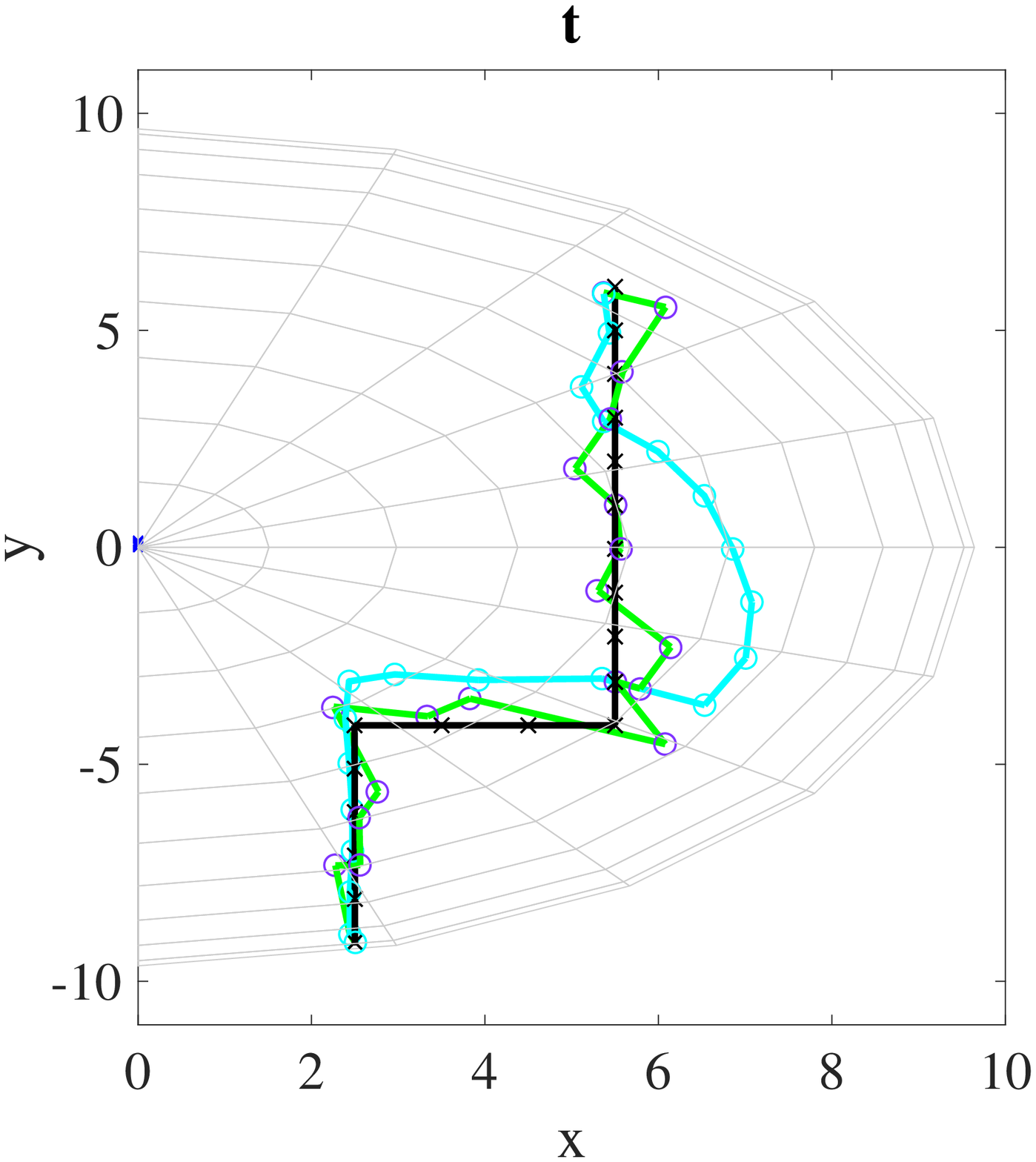}
}
\caption{Tracking estimates for a different trajectory, with transition parameter match, i.e., $\gammat=1$ (left) and mismatch, i.e., $\gammat=10^5$ (right).}
\label{fig:zigzag}
\end{figure}
\subsubsection{Tracking performance}

In Fig.~\ref{fig:trajectories_forCDF}, we
present the estimated trajectories for two different arrays with $N=20\times20$ (top) and $30\times30$ (middle) working with different measurement noise levels ($\sigmaeta=10^\circ$ on the left and $\sigmaeta=20^\circ$ on the right). The Fresnel region is displayed as a grey sphere, whereas the actual trajectory of the source with black with cross markers at each step. The estimated trajectories are depicted with different colors according to the employed approach and to the legend. The parameters used in the models for generating the data are the same used by the estimators, i.e., perfect parameter match ($\gammat=1$, $\gammam=0$). When $N=20\times20$ antennas were used, the initial and final points of the source trajectory were outside the Fresnel region. Consequently, according to the analysis from Sec.~\ref{subsec:dataFIM}, in these areas, measurements are less informative about the source state and larger errors in the trajectory estimation were made, especially by those estimators mainly based on the information retrieved from \ac{LF}, i.e., the \ac{PF} with likelihood \ac{IS} and the \ac{MLE}. 
On the contrary, when operating in the near--field region, a significant tracking performance improvement is obtained under the same measurement noise conditions. 
We notice that the \ac{PF} with linearised optimal \ac{IS} (namely PF - LO-IS in the figures) is less robust and accurate in estimating the trajectory than the other PF methods. In contrast to the prior \ac{IS} (PF - P-IS), the particle propagation depends both on the transition and measurement densities, which are not always in perfect accordance with each other. This is evident when the \ac{LF} becomes extremely peaky (i.e., when the source is very close to the array or with a large number of measurements). Then it is very likely that particles are not propagated in regions of  large probability masses because the likelihood is not overlapped with regions of high 
transition density \cite{bunch2013particle}.
\begin{figure}[t!]
\psfrag{x}[c][c][0.7]{Localization Error [m]}
\psfrag{y}[c][c][0.7]{Empirical CDF}
\psfrag{t2020TM0MM0}[c][c][0.7]{$20\times 20$, Perfect Parameter Match ($\mathsf{TM}_0$, $\mathsf{MM}_0$)}
\psfrag{t3030TM0MM0}[c][c][0.7]{$30\times 30$, Perfect Parameter Match ($\mathsf{TM}_0$, $\mathsf{MM}_0$)}
\psfrag{PEB}[lc][lc][0.6]{$\sqrt{\text{P-CRLB}}$}
\psfrag{data1111111111111111111}[lc][lc][0.6]{Extended Kalman Filter}
\psfrag{ML}[lc][lc][0.6]{MLE}
\psfrag{PFPriorS}[lc][lc][0.6]{\ac{PF}, PF - P-IS}
\psfrag{PFLikelihoodS}[lc][lc][0.6]{PF - L-IS}
\psfrag{PFOptimalS}[lc][lc][0.6]{PF - LO-IS}
\centerline{
\psfrag{t1}[c][c][0.7]{$N=20\times20$}
\includegraphics[width=0.85 \linewidth,draft=false]
{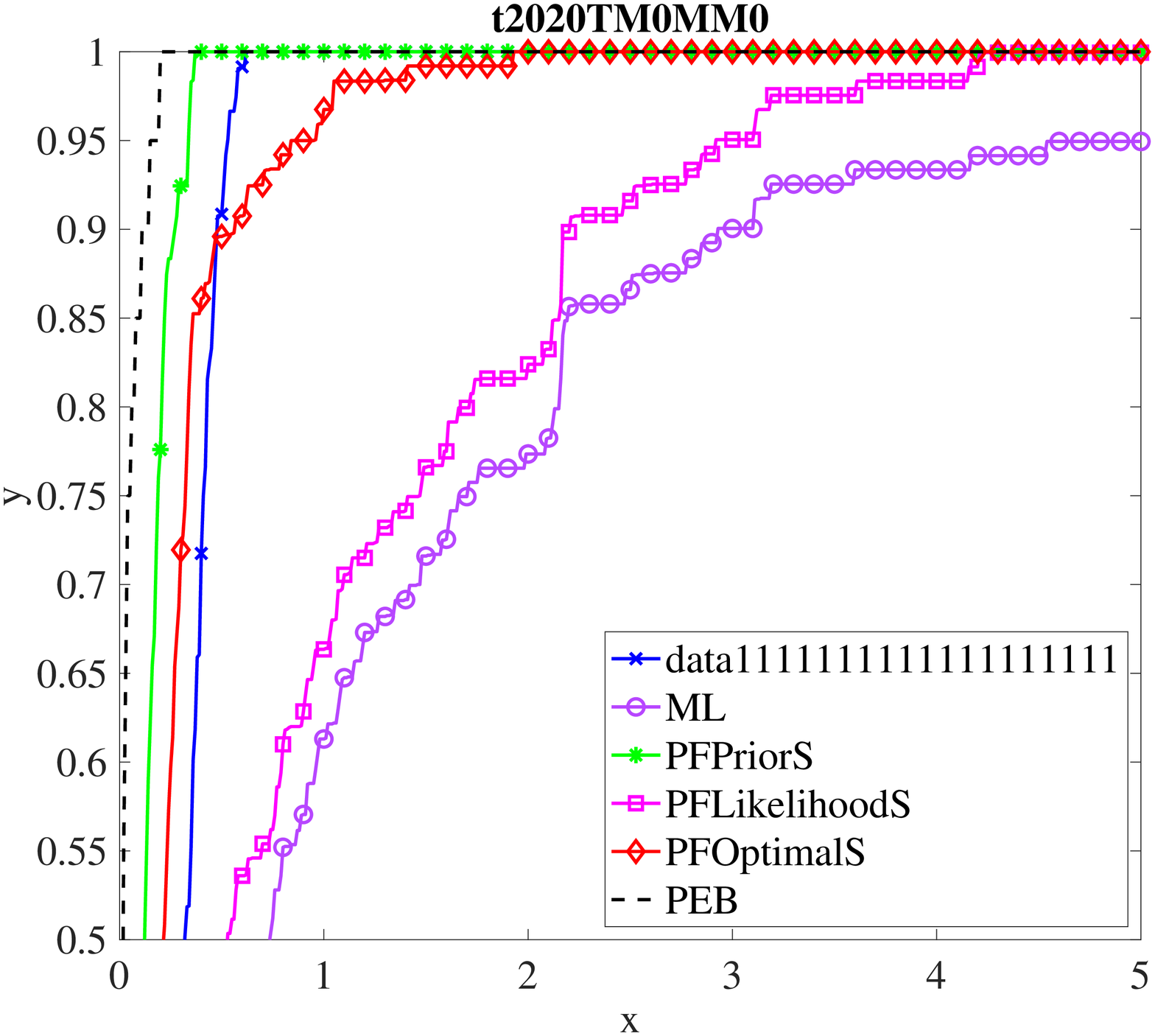}
} 
\vspace{0.5cm}
\centerline{
\psfrag{t1}[c][c][0.7]{$N=30\times30$}
\includegraphics[width=0.85 \linewidth,draft=false]
{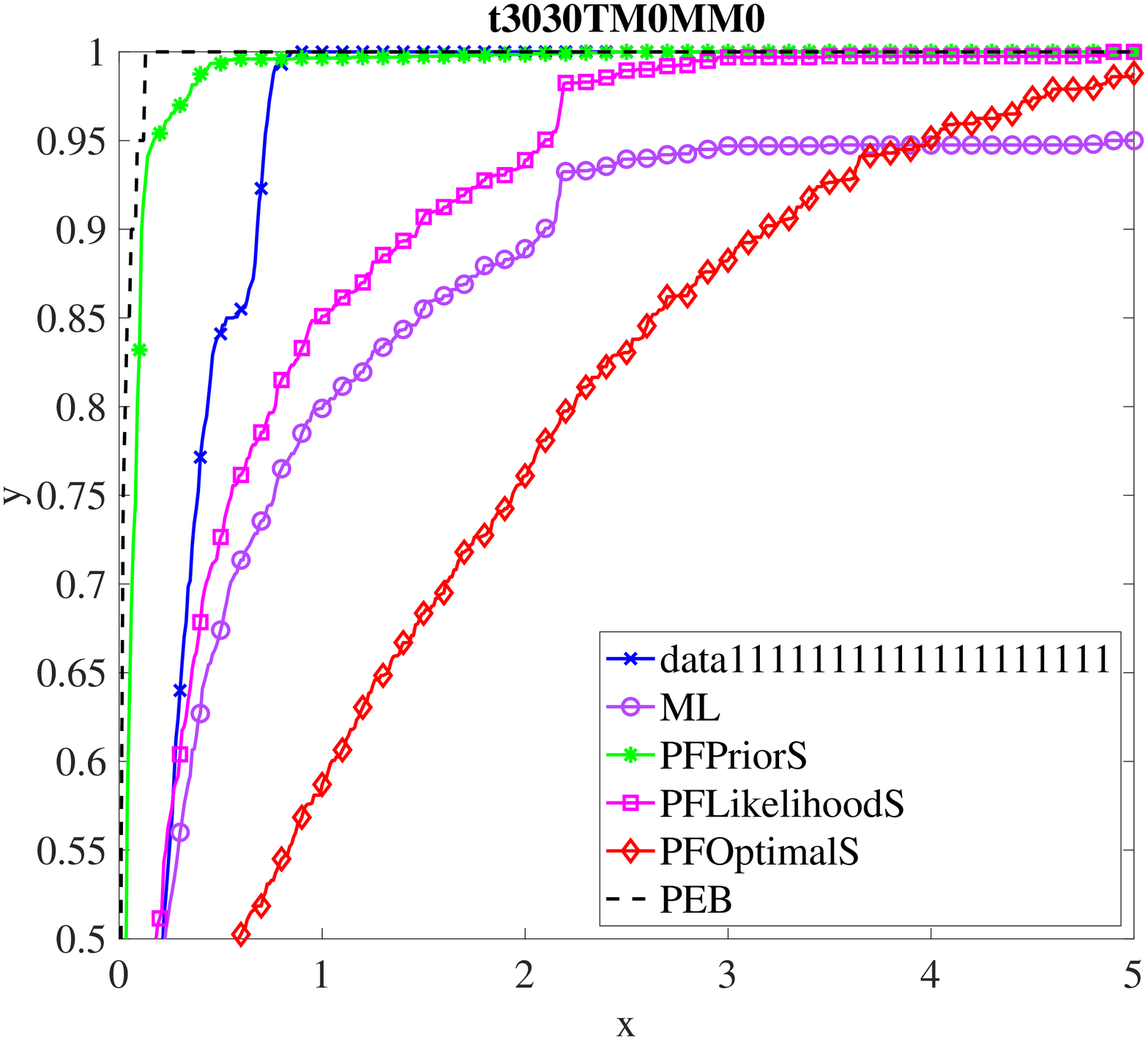}
} 
\caption{Empirical CDF vs. Localization error in meters for different estimators and by considering a rectangular array with $N=20\times 20$ (top) and $N=30\times 30$ (bottom) antennas, respectively. The measurement noise variance is set to $\sigmaeta=20^\circ$. Markers are plotted with a step of 10.}
\label{fig:SuccessRate}
\end{figure}
\begin{figure}[t!]
\psfrag{x}[c][c][0.7]{Localization Error [m]}
\psfrag{y}[c][c][0.7]{Empirical CDF}
\psfrag{t}[c][c][0.7]{Model parameters mismatches}
\psfrag{PFPriorSTM1MM1}[lc][lc][0.6]{Prior IS - TM$_1$ MM$_1$}
\psfrag{PFPriorSTM0MM1}[lc][lc][0.6]{Prior IS - TM$_0$ MM$_1$}
\psfrag{PFPriorSTM0MM0}[lc][lc][0.6]{Prior IS - TM$_0$ MM$_0$}
\psfrag{PFPriorSTM1MM0}[lc][lc][0.6]{Prior IS - TM$_1$ MM$_0$}
\psfrag{PFOptimalSTM1MM1}[lc][lc][0.6]{Optimal IS - TM$_1$ MM$_1$}
\psfrag{PFOptimalSTM0MM1}[lc][lc][0.6]{Optimal IS - TM$_0$ MM$_1$}
\psfrag{PFOptimalSTM0MM0}[lc][lc][0.6]{Optimal IS - TM$_0$ MM$_0$}
\psfrag{PFOptimalSTM1MM0}[lc][lc][0.6]{Optimal IS - TM$_1$ MM$_0$}
\psfrag{PEB}[lc][lc][0.6]{PEB}
\centerline{
\psfrag{t1}[c][c][0.7]{$N=20\times20$, }
\includegraphics[width=0.85 \linewidth,draft=false]
{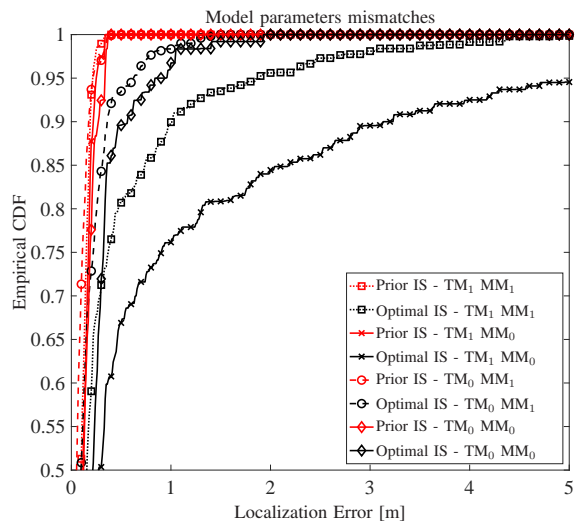}
}
\caption{Empirical CDF vs. Localization error in meters for particle filters using prior (red lines) and optimal (black lines) \ac{IS}, by considering a rectangular array with $N=20\times 20$ antennas. Parameters mismatches are considered as described in Sec.~\ref{sec:casestudy_sp}. Markers are plotted with a step of 10.}
\label{fig:SuccessRate_Comparison}
\end{figure}

We explain this effect in Fig~\ref{fig:weightsMM1}. There we see the weighted particles with linearised optimal and prior \ac{IS} for the measurement update at time instant $k=9$ and time update at $k=8$ (prediction for $k=9$). The color of the particles represents their weights in dB.
During the measurement update, the predictive weights are modified by the likelihood as in \eqref{eq:wup}. In  Fig~\ref{fig:weightsMM1}-(top), the particles of the \ac{PF} linearised optimal \ac{IS} are propagated in a region where the \ac{LF} is not informative and, thus, in this example all the weights of the particles have very low weights of about $-100\,$dB corresponding to the \ac{LF} tails. Moreover, there are two separated clouds of particles because of ambiguities.

Conversely, in the prior \ac{IS} case of Fig~\ref{fig:weightsMM1}-middle, the likelihood peak can be more easily caught because the time update set all the weights at the same probability ($1/M=10^{-3}$, $-30$ dB) and propagate them only using the transition model (with a bigger spatial dispersion with respect to the previous case). In this case, we can see that few particles, indicated with a red arrow in the plot, intercept the peak of the \ac{LF}.

The problem of the optimal \ac{IS} can be partially overcome by increasing the uncertainty on the measurement model, as shown in Figs.~\ref{fig:weightsMM1}-bottom where $\sigmaeta$ is set to $100^\circ$ instead of $10^\circ$ for $N=30 \times 30$. This is equivalent to perform a \textit{roughening} operation, i.e., a spreading of the \ac{LF} by increasing its variance. 
Figure~\ref{fig:trajectories_forCDF}-bottom shows the two estimated trajectories by considering an augmented  measurement variance.

In Fig.~\ref{fig:zigzag}, we show a different target trajectory with abrupt changes in direction. Indeed, rapid variations of the target trajectory are more challenging from a tracking perspective \cite{bugallo2007performance}. Consequently, in this case, having a parameter model mismatch is beneficial in order to increase the probability of propagating particles in informative transition regions. In fact, the target trajectory is not well described by the model in \eqref{eq:tmodel} and, thus, a bigger covariance matrix leads to smaller
inertia in the estimation process and to better results.

The previous results were obtained by considering a single realization in order to get a qualitative idea about the performance behaviour of the investigated tracking algorithms and of their robustness. Now, considering the same target trajectory shown in Fig.~\ref{fig:trajectories_forCDF}, we perform a performance comparison through Monte Carlo simulations of many realizations of trajectories. 
As a metric for comparison, we consider the empirical \ac{CDF} defined as
\begin{align}
    &\mathsf{CDF}\left( e_\mathsf{th}\right)= \frac{1}{\Nmc\, K}  \sum_{\ell=1}^{\Nmc} \sum_{k=1}^{K} \mathbf{1}\left(e_{\ell,k} \le e_\mathsf{th} \right) \,,
\end{align}
where the Monte Carlo cycles were fixed to $\Nmc=100$,  $\mathbf{1}\left(\cdot\right)$ is equal to one if its logical argument is true, otherwise it is zero, $e_{\ell,k}=\lVert \hat{\mathbf{p}}_{\ell,k} - \mathbf{p}_k \rVert$, $\hat{\mathbf{p}}_{\ell,k}$ is the estimated target position at time instant $k$ for the $\ell$th Monte Carlo run, $\mathbf{p}_k$ is the actual target position at the $k$th time instant and $e_\mathsf{th}$ is a threshold for the localization error.

Figure~\ref{fig:SuccessRate} depicts the \ac{CDF} obtained for $N=20\times 20$ (top) and $N=30\times 30$ (bottom), respectively, for $\sigmaeta=20^\circ$ 
when the parameters match both in the measurement and transition models (i.e., $\mathsf{TM}_0$ and $\mathsf{MM}_0$). 
The $\sqrt{\text{P-CRLB}}$ is also shown as performance benchmark. 
As intuitively predictable,
the likelihood \ac{IS} performs better for $30\times 30$ than for the $20\times 20$ thanks to the more peaky \ac{LF}, as the target is located always within the near--field region of the receiver. On the other hand, the \ac{PF} with the optimal \ac{IS} with linearised observation model has lower performance, especially for the $30\,\times\,30$ array. The \ac{EKF} also allows to attain reliable performance despite its low complexity, provided that it is well initialized.
We also evaluated the impact of parameter mismatches for the \ac{PF}, considering a $20\times 20$ array. 
The results in Fig.~\ref{fig:SuccessRate_Comparison} suggest that the prior \ac{IS} is robust to model mismatches, as the red curves exhibit similar behaviors. In particular, with large variances in the models, the system was more robust to trajectory variations and, consequently, it could track the target with a slightly improved accuracy. On the flip side, the optimal \ac{IS} with linearized likelihood was more sensitive to the accuracy of the model.
In this case, the joint variations in the transition model together with the peaky likelihood dramatically affected the performance. We observed performance improvement by introducing a measurement mismatch ($\text{TM}_0$, $\text{MM}_1$). 
Finally, the results of the performance as a function of the number of particles are presented in Fig.~\ref{fig:Nparticles}. They show that $M=1000$ is a good tradeoff in terms of obtained accuracy. 
\begin{figure}[t!]
\psfrag{x}[c][c][0.7]{Localization Error [m]}
\psfrag{y}[c][c][0.7]{Empirical CDF}
\psfrag{t1}[c][c][0.7]{Number of particles}
\psfrag{data1111111111111111111}[lc][lc][0.6]{Prior IS - $M=10000$}
\psfrag{PFPriorS10000aaaaaa}[lc][lc][0.6]{Prior IS - $M=10000$}
\psfrag{PFPriorS1000}[lc][lc][0.6]{Prior IS - $M=1000$}
\psfrag{PFPriorS100}[lc][lc][0.6]{Prior IS - $M=100$}
\psfrag{PFOptimalS100}[lc][lc][0.6]{Optimal IS - $M=100$}
\psfrag{PFOptimalS1000}[lc][lc][0.6]{Optimal IS - $M=1000$}
\psfrag{PFOptimalS10000}[lc][lc][0.6]{Optimal IS - $M=10000$}
\centerline{
\includegraphics[width=0.85 \linewidth,draft=false]
{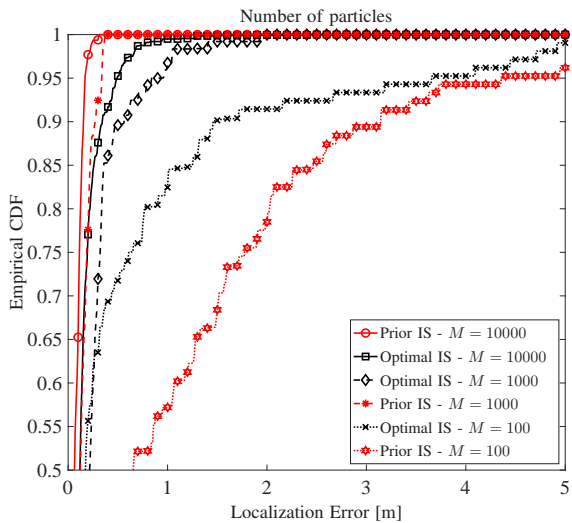}
}
\caption{Empirical CDF vs. Localization error in meters for particle filters using prior (red lines) and optimal (black lines) \ac{IS}, by considering a rectangular array with $N=20\times 20$ antennas. The number of particles spans from $100$ to $10000$. Markers are plotted with a step of 10.}
\label{fig:Nparticles}
\end{figure}

\section{Conclusions}\label{sec:conclusions}
In this paper, we investigated a tracking problem where a single array equipped with a co-located large number of antennas estimates the position of a source by exploiting the \ac{CoA} information. 
First, we derived the theoretical bound on tracking estimation error and we investigated the capability of the system to infer both angle and distance information (i.e., the position), when operating in the near--field region, thanks to the exploitation of the electromagnetic wavefront curvature. The asymptotic analysis puts in evidence that the distance information tends to vanish when approaching the far--field region, which implies a scarce position estimation in the radial direction. 
Second, we compared the performance of some state-of-the-art practical tracking algorithms to show the feasibility of accurate tracking using \ac{CoA} information under different working conditions.
Numerical results show that the performance of \ac{PF}-based schemes is close to that of the theoretical bound, and hence sub-meter accuracy can be obtained in the considered scenarios. Moreover, the comparison among different methods highlights  that \ac{PF} with prior \ac{IS} is more robust to model mismatches and to abrupt trajectory changes.
Finally, our study indicates that it is possible to perform high accuracy tracking using only one single antenna array and narrowband signals by exploiting the \ac{CoA}, provided that the target is within the near--field region of the antenna. Operating like this, there is no need of accurate \ac{TOA} estimation, which requires very large bandwidths and tight synchronization between the transmitter and the receiver.

\section*{Appendix A}
We derive the expression for $\nabla_{\ssk}\, \Delta \dnk$ where, for simplicity of notation, we omit the temporal index $k$. In particular, due to the fact that the derivatives made with respect to the source velocity $\vsk$ are $0$, we focus only on the source position. We can write
\begin{align}\label{eq:nablas2}
   \!\!\nabla_{\ps} \,\Delta \dn &= \left(\nabla_{\ps} d \right) \! \left[\sqrt{f_n}-1 \right]+ \frac{d \,\, \nabla_{\ps}\,\, f_n}{2\, \sqrt{f_n}}, 
\end{align}
where $\nabla_{\ps}\left( f_n\right)$ is given by
\begin{align}\label{eq:nablas3}
   \nabla_{\ps} f_n =& - \frac{2\, \dno}{d} \left( \frac{\dno\, \nabla_{\ps} d}{d^2} + \nabla_{\ps} g_n - \frac{g_n \, \nabla_{\ps}\,d}{d} \right),
\end{align}
with 
\begin{align}
\nabla_{\ps}\!\left( d \right)&\!\!=\!\left[  \frac{\partial d}{\partial x}, \frac{\partial d}{\partial y}, \frac{\partial d}{\partial z} \right] 
= \left[  \frac{\xs-\xo}{d},\, \frac{\ys-\yo}{d},\, \frac{\zs-\zo}{d} \right],
\end{align}
being the gradient of the distance with respect to the source position and where the gradient of the angular term $g_n$ is 
\begin{align}\label{eq:nablas5}
    \!\!\!\! &\nabla_{\ps}\, g_n=  \sin(\thetano) \left( \cos\left(\phino - \phi \right) \cos\left( \theta \right) \nabla_{\ps}\left( \theta \right)  \right.    \\
   & \left. \!\! + \sin\left(\theta \right) \sin \left(\phino\!-\!\phi \right)\! \nabla_{\ps} \left(\phi \right) \right) \!-\! \cos(\thetano)\sin \left(\theta \right) \! \nabla_{\ps} \!\left(\theta \right), \nonumber
\end{align}
with
\begin{align}\label{eq:dtheta}
    \nabla_{\ps}(\theta)&=\left[  \frac{\partial \theta}{\partial x},\, \frac{\partial \theta}{\partial y},\, \frac{\partial \theta}{\partial z} \right] \nonumber \\
    &=\!\left[\frac{\cos(\phi)\cos(\theta)}{d},\frac{\sin(\phi)\cos(\theta)}{d},-\frac{\sin(\theta)}{d} \right]\,, \\ \nabla_{\ps}(\phi)&\! =\!\left[  \frac{\partial \phi}{\partial x},\, \frac{\partial \phi}{\partial y},\, \frac{\partial \phi}{\partial z} \right] \!=\!\left[-\frac{\sin(\phi)}{d\,\sin(\theta)},\, \frac{\cos(\phi)}{d\,\sin(\theta)},\,0 \right]. \label{eq:dphi}
\end{align}

Now, it is easy to show that \eqref{eq:limite} holds. In fact, considering for example the $x$ coordinate, we have 
\begin{align}
   \frac{\partial \Delta \dn}{\partial x} &= \frac{1}{d} \left\{(\xs-\xo)\! \left[\sqrt{f_n}-1 \right]- \frac{ \dno}{\sqrt{f_n}} \left( \frac{\dno\, (\xs-\xo)}{d^2} \right. \right. \nonumber \\
   &\left.\left.+  d\, \frac{\partial g_n}{\partial x} - \frac{g_n \, (\xs-\xo)}{d} \right)\right\} \underset{\dk \gg \dF}{\longrightarrow} 0.
\end{align}

\section*{Appendix B}

By substituting $\frac{\partial h_n }{\partial \xi}  =\frac{2\, \pi}{\lambdap}\, \frac{\partial \Delta d_n}{\partial \xi}$ and by omitting the temporal index for notation simplicity, we can reformulate \eqref{eq:dfim}  according to \cite{jourdan2008position} as
\begin{align}\label{eq:Jkparam}
\tjM\left(\xi\right)
&= \frac{1}{\sigmaeta^2} \, \sumn\left( \frac{\partial h_n}{\partial \xi} \right)^2 \!\!= \frac{4\, \pi^2}{\lambdap^2\,\sigmaeta^2} \, \sumn\left( \frac{\partial \Delta d_n}{\partial \xi} \right)^2, 
\end{align}
where the derivatives inside the summations of \eqref{eq:Jkparam} depend on the actual significance of $\xi$ and are given by
\begin{align}\label{eq:d1d}
&\frac{\partial \Delta \dn}{\partial d} = \sqrt{\fn}-1-\frac{\dno}{d\,\sqrt{\fn}}\left(\frac{\dno}{d} - \gn \right), \\
&\frac{\partial \Delta \dn}{\partial \theta}= d\, \frac{\partial \sqrt{\fn}}{\partial \theta} = -\frac{\dno}{d\sqrt{\fn}} \, \frac{\partial \gn}{\partial \theta}, \\
&\frac{\partial \Delta \dn}{\partial \phi}= d\, \frac{\partial \sqrt{\fn}}{\partial \phi} = -\frac{\dno}{d\sqrt{\fn}} \, \frac{\partial \gn}{\partial \phi} \,, \label{eq:d1p}
\end{align}
with $\fn$ and $\gn$ defined in \eqref{eq:f} and \eqref{eq:gn}, respectively, and with
\begin{align}\label{eq:gnd}
    &\frac{\partial \gn}{\partial \theta}\!=\!\cos(\theta) \sin(\thetano) \cos(\phino-\phi)\!-\!\sin(\theta)\cos(\thetano),  \\
    &\frac{\partial \gn}{\partial \phi}\!=\!\sin(\phino-\phi)\, \sin(\thetano)\, \sin(\theta).
\end{align}
By substituting $f_n$ in \eqref{eq:d1d}-\eqref{eq:d1p} and squaring, we obtain
\begin{align} \label{eq:functions}
   &\left(\frac{\partial \Delta \dn}{\partial d} \right)^2  = \frac{\left(1- \frac{g_n\,\dno}{d} - \sqrt{1+ \frac{\dno^2}{d^2}-2\, \frac{g_n\,\dno}{d}} \right)^2}{1+ \frac{\dno^2}{d^2}-2\, \frac{g_n\,\dno}{d}},
   \\
  &\left(\frac{\partial \Delta \dn}{\partial \theta} \right)^2  =  \frac{\dno^2}{1+ \frac{\dno^2}{d^2}-2\, \frac{g_n\,\dno}{d}} {\left(\frac{\partial \gn}{\partial \theta} \right)^2}, \\
  &\left(\frac{\partial \Delta \dn}{\partial \phi} \right)^2  =  \frac{\dno^2}{1+ \frac{\dno^2}{d^2}-2\, \frac{g_n\,\dno}{d}} {\left(\frac{\partial \gn}{\partial \phi} \right)^2}.
  \label{eq:functions2}
\end{align}   
By injecting the expressions above into \eqref{eq:Jkparam}, we obtain \eqref{eq:Jkdfinal}. 
%
%
%
%

\section*{Appendix C}
In this case, given {\it A1}, we can write \eqref{eq:Jkdfinal1}-\eqref{eq:Jkthetafinal1} as
\begin{align}\label{eq:CFIMd1}
&\tjM\left(d \right) =\frac{4\, \pi^2}{\lambdap^2\, \sigmaeta^2} \sumn \frac{8\, d^2 +D^2-4\, d\, \sqrt{4\,d^2+D^2}}{4\,d^2+D^2} \\
&\tjM\left(\theta \right)  =\frac{\pi^2}{\lambdap^2 \sigmaeta^2}\, \frac{D^2}{1+ \frac{D^2}{4\, d^2}}\, \sumn \left(\cos \left(2\, \pi \frac{n}{N}\right) \right)^2 \\
&\tjM\left(\phi \right)  =\frac{\pi^2}{\lambdap^2 \sigmaeta^2}\, \frac{D^2}{1+ \frac{D^2}{4\, d^2}}\, \sumn \left(\sin \left(2\, \pi \frac{n}{N}\right) \right)^2 
\label{eq:CFIMp1}
\end{align}
where we have exploited the following relationships
\begin{align}
&\left(\frac{\partial \gn}{\partial \theta}\right)^2=\left(\cos \thetano \right)^2=\left(\cos \left(2\, \pi \frac{n}{N}\right) \right)^2\\
&\left(\frac{\partial \gn}{\partial \phi}\right)^2=\left(\sin \thetano\right)^2=\left(\sin \left(2\, \pi \frac{n}{N}\right)\right)^2.
\end{align}
Then, \eqref{eq:CFIMd1}-\eqref{eq:CFIMp1} can be further simplified as in \eqref{eq:CFIMd}-\eqref{eq:CFIMp}.

%


\end{document}